\documentclass[12pt,final,
thmsa,
epsf,a4paper]
{article}
\usepackage{booktabs}
\newcommand{\ra}[1]{\renewcommand{\arraystretch}{#1}}

\newtheorem{theorem}{Theorem}

\newtheorem{corollary}{Corollary}
\newtheorem{definition}{Definition}
\newtheorem{example}{Example}
\newtheorem{lemma}{Lemma}
\newtheorem{proposition}{Proposition}
\newtheorem{remark}{Remark}

\newenvironment{proof}[1][Proof]{\emph{#1.} }{\  \hfill $\square $ \vspace{5 pt}}

\usepackage{natbib}
\usepackage{amssymb}
\usepackage[dvips]{graphics}
\usepackage{amsmath}

\usepackage{mathpazo}

\usepackage{enumerate}

\usepackage{amsfonts}

\usepackage{amssymb}

\usepackage{enumerate}

\usepackage{mathrsfs}

\usepackage[usenames]{color}

\usepackage{pgf,tikz}
\usetikzlibrary{arrows}

\usetikzlibrary{decorations.pathreplacing,angles,quotes}
\usetikzlibrary{arrows.meta}
\usetikzlibrary{arrows, decorations.markings}
\tikzset{myptr/.style={decoration={markings,mark=at position 1 with %
       {\arrow[scale=2,>=stealth]{>}}},postaction={decorate}}}

\usepackage{hyperref}
\hypersetup{
    colorlinks,
    citecolor=blue,
    filecolor=black,
    linkcolor=blue,
    urlcolor=blue
}

\usepackage[utf8]{inputenc}
\usepackage{amssymb, amsmath,geometry}
\usepackage{fancyhdr}

\usepackage{soul}

\usepackage{emptypage}

\newcommand*\samethanks[1][\value{footnote}]{\footnotemark[#1]}

\usepackage[T1]{fontenc}
\DeclareFontFamily{T1}{calligra}{}
\DeclareFontShape{T1}{calligra}{m}{n}{<->s*[1.44]callig15}{}
\DeclareMathAlphabet\mathcalligra   {T1}{calligra} {m} {n}

\begin{document}

\title{A characterization of absorbing sets in coalition formation games\thanks{We thank the Advisory Editor and two anonymous reviewers for their valuable comments. We also thank Elena Molis for her participation in the early stages of the paper, and  Coralio Ballester, Jordi Massó, Alejandro Neme, Oscar Volij,  and Peio Zuazo-Garin for their suggestions and comments. Inarra acknowledges financial support from the Spanish Ministry of Economy and Competitiveness (project PID2019-107539GB-I00), and from the Basque Government (project IT1367-19); and the hospitality of the Instituto de Matemática Aplicada San Luis, Argentina. Bonifacio and  Neme acknowledge financial support
from the  UNSL through grants 032016 and  030320,  from the Consejo Nacional
de Investigaciones Cient\'{\i}ficas y T\'{e}cnicas (CONICET) through grant
PIP 112-200801-00655, and from Agencia Nacional de Promoción Cient\'ifica y Tecnológica through grant PICT 2017-2355; and the hospitality of  
the University of the Basque Country (UPV/EHU), Spain, during their respective visits.}}

\author{A. G. Bonifacio\thanks{
Instituto de Matem\'{a}tica Aplicada San Luis (UNSL-CONICET), Departamento de Matemática (Universidad Nacional de San Luis),  Argentina. Emails: \texttt{abonifacio@unsl.edu.ar} 
and \texttt{paneme@unsl.edu.ar} 
} \and E. Inarra\thanks{
Department of Economic Analysis  and Public Economic Institute, University of the
Basque Country (UPV/EHU), 48015 Bilbao, Spain. Email: \texttt{elena.inarra@ehu.eus}}
\and P. Neme\samethanks[2]}

\maketitle

\abstract{
Given a standard myopic dynamic process among coalition structures, an absorbing set is a minimal collection of such structures that is never left once entered through that process. Absorbing sets are an important solution concept in coalition formation games, but they have drawbacks: they can be large and hard to obtain. In this paper, we characterize an absorbing set in terms of a collection consisting of a small number of sets of coalitions that we refer to as a ``reduced form'' of a game. We apply our characterization to study convergence to stability in several economic environments.
\bigskip




}%

\maketitle
\section{Introduction}

Coalition formation games have become a central focus in a substantial body of literature addressing diverse social and economic issues. This includes topics such as the establishment of cartels, lobbying, customs unions, conflicts, the provision of public goods, and the formation of political parties  \citep[see][]{ray2007game,ray2015coalition}. Coalition formation games encompass matching models that have significantly improved the understanding and design of mechanisms for various applications, such as school choice and kidney exchange  \citep[see][and references therein]{roth2018marketplaces}, one-sided problems like the roommate problem, as well as two-sided problems ranging from the classical marriage problem to many-to-one matching problems with peer effects and externalities.  

A widely studied solution concept for coalition formation games is that of
 a stable coalition structure.\footnote{In the literature on coalition formation games the papers by \cite{banerjee2001core,bogomolnaia2002stability,iehle2007core} identify structures of preferences that guarantee the existence of stable coalition structures.  \cite{echenique2007solution} develop an algorithm for matching markets with preferences over colleagues to determine the existence of stable matchings. \cite{pycia2012stability} and \cite{gallo2018rationing}, in different contexts, study what sharing rules induce stable coalition formation games. Two other notable contributions to studying sharing rules and stable coalition structures are \cite{barbera2015meritocracy} and \cite{herings2021last}.}  A coalition structure is a partition of the set of agents into coalitions and a stable coalition structure is a coalition structure that  "protects" its coalitions in the following way: Whenever some agents have an incentive to deviate to an external coalition (outside the coalition structure), another coalition in the structure impedes its formation because there is (at least) one agent that prefers that coalition to the external coalition.  Note, however, that a coalition formation game may have no stable coalition structure.

In coalition formation games, there is a more general solution concept called absorbing set.\footnote{This notion has been studied by several authors under different names and in different contexts. As far as we know, \cite{schwartz1970possibility} was the first to introduce it for collective decision-making problems. See also  \cite{inarra2013absorbing,jackson2002evolution,olaizola2014asymmetric}. The union of absorbing sets gives the ``admissible set'' \citep{kalai1977admissible}. Recently, \cite{demuynck2019myopic} define the ``myopic stable set'' in a very general class of social environments and study its relation to other solution concepts. } An important feature of these sets is that they are known to exist for every coalition formation game. Consider the dynamic process where an unstable coalition structure undergoes adjustments when certain agents collectively decide to deviate and form a new coalition, with abandoned agents becoming singletons in the new structure.\footnote{This process corresponds to the $\gamma$-model of \cite{hart1983endogenous}. These authors argue that when a coalition is formed through the agreement of all its members, and subsequently some agents withdraw, the agreement dissolves, and the remaining agents become singletons. In our analysis, this assumption fits well, as our modeling may not encompass all possible coalitions, and in such instances, the agents who are abandoned may not be able to remain as a ``permissible'' coalition.} An absorbing set is a minimal collection of coalition structures that, once entered throughout this dynamic process, is never left. Furthermore, a stable coalition structure generates an absorbing set consisting of only that coalition structure. 
We call ``trivial'' to these absorbing sets. ``Non-trivial'' absorbing sets, in counter-position, have several coalition structures, none of which are stable. 

In the case of roommate problems, coalition formation is well-understood. A roommate problem has either trivial absorbing sets or non-trivial absorbing sets \citep{inarra2013absorbing}. The existence of non-trivial absorbing sets in these problems is equivalent to the existence of a set consisting of an odd number of coalitions presenting cyclical behavior, called an ``odd ring'' \citep{tan1991necessary}. Furthermore, in the trivial absorbing set case,  starting from an arbitrary coalition structure a path to a stable coalition structure is guaranteed \citep{diamantoudi2004random}. However, none of these results hold in more general settings where trivial and non-trivial absorbing sets can coexist.

Absorbing sets in general settings have other drawbacks. On one hand, in order to compute an absorbing set, every coalition structure must be computed. On the other hand, in order to define them, a dynamic process among coalition structures has to be specified. 
To avoid such drawbacks, our goal is to characterize an absorbing set in terms of a collection consisting of a small number of sets of coalitions, which we call a ``reduced form''. This reduced form condenses all the relevant information on what coalition structures will form in the absorbing set.

To illustrate the complexity of an absorbing set consider the game presented in Table \ref{table para la reduced form}.\footnote{In order to ease notation, throughout the paper we denote coalitions without curly brackets and commas, i.e., coalition $\{6,7,8\}$ is simply written $678$.}

\begin{table}[h!]{\small
    \centering
         \begin{tabular}{@{}ccccccccccccccccccc@{}}\hline
        $\boldsymbol{1}$ && $\boldsymbol{2}$ && $\boldsymbol{3}$ && $\boldsymbol{4}$ && $\boldsymbol{5}$ && $\boldsymbol{6}$ && $\boldsymbol{7}$ && $\boldsymbol{8}$ \\ 
        \hline
        $12$ && $23$ && $34$ && $46$ && $45$ && $678$ && $78$ && $78$\\ 
        $13$ && $12$ && $13$ && $45$  && $56$ && $56$ && $678$ && $678$\\ 
        $1$&& $2$ && $23$ &&$34$ &&$5$ &&$46$ && $7$ && $8$  \\ 
          &&  && $3$ &&$4$ &&   && $6$ && &&    \\ \hline
         \end{tabular} 
         \caption{\small An 8-agent game.}
          \label{table para la reduced form}}
\end{table}

\noindent Let us take a look at the preferences in Table \ref{table para la reduced form}. Agent $1$, for instance, prefers coalition $12$ to coalition $13$, and no other coalition (to which agent $1$ belongs) is ``permissible'' for this agent. It can be checked that this game has only one absorbing set, depicted in Figure \ref{figura intro}. This absorbing set is non-trivial and has fourteen coalition structures. Note that coalition $78$ is present in every coalition structure of the absorbing set.

 \begin{figure}[h!]
 \centering
\begin{tikzpicture}[scale=1]

\node (1) at (6,12)  {\scriptsize{$\{1,2,3,45,6,78\}$}};
\node (2) at (9,11)  {\scriptsize{$\{1,2,3,46,5,78\}$}};
\node (3) at (11,9)  {\scriptsize{$\{1,2,3,4,56,78\}$}};
\node (4) at (12,7)  {\scriptsize{$\{12,3,45,6,78\}$}};
\node (5) at (12,5)  {\scriptsize{$\{1,23,45,6,78\}$}};
\node (6) at (11,3) {\scriptsize{$\{13,2,45,6,78\}$}};
\node (7) at (9,1)  {\scriptsize{$\{12,3,46,5,78\}$}};
\node (8) at (6,0)  {\scriptsize{$\{1,23,46,5,78\}$}};
\node (9) at (3,1)  {\scriptsize{$\{13,2,46,5,78\}$}};
\node (10) at (1,3)  {\scriptsize{$\{12,3,4,56,78\}$}};
\node (11) at (0,5)  {\scriptsize{$\{1,23,4,56,78\}$}};
\node (12) at (0,7) {\scriptsize{$\{13,2,4,56,78\}$}};
\node (13) at (1,9)  {\scriptsize{$\{1,2,34,56,78\}$}};
\node (14) at (3,11)  {\scriptsize{$\{12,34,56,78\}$}};

\draw[blue,-{Latex[length=1.3mm]}] (3) to (1);
\draw[blue,-{Latex[length=1.3mm]}] (13) to (1);
\draw[blue,-{Latex[length=1.3mm]}] (1) to (2);
\draw[blue,-{Latex[length=1.3mm]}] (2) to (3);
\draw[blue,-{Latex[length=1.3mm]}] (1) to (4);
\draw[blue,-{Latex[length=1.3mm]}] (6) to [bend left](4);
\draw[blue,-{Latex[length=1.3mm]}] (10) to (4);
\draw[blue,-{Latex[length=1.3mm]}] (14) to (4);
\draw[blue,-{Latex[length=1.3mm]}] (1) to (5);
\draw[blue,-{Latex[length=1.3mm]}] (11) to (5);
\draw[blue,-{Latex[length=1.3mm]}] (4) to (5);
\draw[blue,-{Latex[length=1.3mm]}] (1) to (6);
\draw[blue,-{Latex[length=1.3mm]}] (12) to (6);
\draw[blue,-{Latex[length=1.3mm]}] (5) to (6);
\draw[blue,-{Latex[length=1.3mm]}] (9) to (7);
\draw[blue,-{Latex[length=1.3mm]}] (4) to [bend right] (7);
\draw[blue,-{Latex[length=1.3mm]}] (2) to [bend right] (7);
\draw[blue,-{Latex[length=1.3mm]}] (7) to (8);
\draw[blue,-{Latex[length=1.3mm]}] (2) to [bend right] (8);
\draw[blue,-{Latex[length=1.3mm]}] (5) to (8);
\draw[blue,-{Latex[length=1.3mm]}] (8) to (9);
\draw[blue,-{Latex[length=1.3mm]}] (6) to (9);
\draw[blue,-{Latex[length=1.3mm]}] (2) to [bend right]  (9);
\draw[blue,-{Latex[length=1.3mm]}] (3) to (10);
\draw[blue,-{Latex[length=1.3mm]}] (12) to [bend left] (10);
\draw[blue,-{Latex[length=1.3mm]}] (7) to (10);
\draw[blue,-{Latex[length=1.3mm]}] (3) to (11);
\draw[blue,-{Latex[length=1.3mm]}] (10) to (11);
\draw[blue,-{Latex[length=1.3mm]}] (8) to (11);
\draw[blue,-{Latex[length=1.3mm]}] (3) to (12);
\draw[blue,-{Latex[length=1.3mm]}] (11) to (12);
\draw[blue,-{Latex[length=1.3mm]}] (9) to  (12);
\draw[blue,-{Latex[length=1.3mm]}] (3) to (13);
\draw[blue,-{Latex[length=1.3mm]}] (12) to (13);
\draw[blue,-{Latex[length=1.3mm]}] (11) to [bend right] (13);
\draw[blue,-{Latex[length=1.3mm]}] (10) to [bend right] (14);
\draw[blue,-{Latex[length=1.3mm]}] (13) to (14);
\end{tikzpicture}
\caption{\small Non-trivial absorbing set of game.}
\label{figura intro}
\end{figure}

In Table \ref{table para la reduced form},  we can also observe that the set of coalitions $\{13,12,23\}$ presents ``cyclical behavior'': agent $1$ prefers coalition $12$ to coalition $13$, agent $3$ prefers coalition $13$ to coalition $23$, and agent $2$ prefers coalition $23$ to coalition $12$. Observe that this cyclical behavior among coalitions induces a cyclical behavior among coalition structures. For instance, in Figure \ref{figura intro}, $\{12,3,46,5,78\}$ dominates $\{13,2,46,5,78\}$, $\{13,2,46,5,78\}$ dominates $\{1,23,46,5,78\}$, and in turn $\{1,23,46,5,78\}$ dominates $\{12,3,46,5,78\}$. The same occurs with the set of coalitions $\{45,46,56\}$ and, for instance, with coalition structures $\{12,3,45,6,78\}$, $\{12,3,46,5,78\}$, and $\{12,3,4,56,78\}$. Notice that, although both sets of coalitions $\{13,12,23\}$ and $\{45,46,56\}$ present cyclical behavior, their influence on the absorbing set is different since:
\begin{enumerate}[(i)]
    \item in each coalition structure of the absorbing set, there is one coalition of the set $\{45,46,56\}$, and
    \item some coalition structures have no coalition of the set $\{13,12,23\}$ present.
   
\end{enumerate}
The only deviations from coalitions in the set $\{45,46,56\}$ that agents can perform within the absorbing set involve only coalitions in that same set. To see this, notice that the only agent having incentive to deviate from any coalition structure of the absorbing set to an external coalition is agent $6$ (to coalition $678$). As coalition $78$ is formed in each coalition structure of the absorbing set, and agent  $7$ prefers coalition $78$ to coalition $678$, such deviation will never be performed. So the cyclical behavior of the set of coalitions $\{45,46,56\}$ is never disrupted. However, this is not true for the set  $\{13,12,23\}$. For instance, agent $3$  has the incentive to deviate from coalition structure $\{13,2,4,56,78\}$ to the external coalition $34$, and there is no way to impede such deviation. When this happens, the cyclical behavior of the set of coalitions $\{13,12,23\}$ is disrupted.

Now, consider the collection formed by the following three sets:  $\{45,46,56\}$, $\{78\}$, and $\{1,2,3\}$. We claim that this collection condenses all the relevant information to reconstruct the absorbing set. As just analyzed, sets $\{45,46,56\}$, and $\{78\}$ ``protect'' each other from external deviations, implying (i) and the fact that coalition $78$ is present in every coalition structure of the absorbing set.  Furthermore, some coalition structures in the absorbing set have all agents of the third set as singletons, implying (ii).
Collections of this type, which we call ``reduced forms'', are our proposal to characterize absorbing sets. They are composed by
\begin{itemize}
    \item sets consisting of one non-singleton coalition that we call ``fixed components'', 
    \item sets consisting of coalitions presenting cyclical behavior that we call ``generalized rings'', and
    \item a set consisting of singleton coalitions.
\end{itemize}
Fixed components and generalized rings in a reduced form protect each other in a stable way, whereas any fixed component or generalized ring formed by agents of the set of singletons is \emph{not} protected by this collection. 
Our main finding is that each absorbing set generates a reduced form and, conversely, each reduced form generates an absorbing set. Thus, by knowing the reduced form, any coalition structure of the absorbing set can be identified.\footnote{A typical coalition structure of the absorbing set of the previous example has the coalition of the fixed component $\{78\}$, a coalition of the generalized ring $\{45,46,56\}$, and permissible coalitions formed by the rest of the agents. For instance, \{1,2,3,45,6,78\} or \{12,34,56,78\}.} These coalition structures include the coalition of each fixed component, some coalitions of the generalized rings, and some permissible coalitions that agents of the set of singletons belong to.

A reduced form with no generalized rings can be identified with a trivial absorbing set. Therefore, as a by-product of our main result,  we find that a coalition formation game having a reduced form with no generalized ring has at least one stable coalition structure.

The essence of generalized ring and reduced form notions have already been identified in roommate problems. On the one hand,  a generalized ring adapts the notion of an odd ring to general coalition formation games. On the other hand, the idea of reduced form has been studied in roommate problems under the name of ``maximal stable partitions'' \citep[][]{inarra2013absorbing}. An important difference between roommate problems and general coalition formation games is that the existence of a generalized ring itself does not imply either the non-existence of trivial absorbing sets or the existence of a reduced form to which that generalized ring belongs.\footnote{Notice that, in the game of Table \ref{table para la reduced form}, generalized ring $\{13,12,23\}$ does not belong to the unique reduced form.}

 A pertinent question for coalition formation games with a stable coalition structure is whether agents always reach stability when left to interact on their own, or whether there is a need to introduce an arbitrator so that agents can achieve stability. We say that a game exhibits convergence to stability if starting from any coalition structure, there is a path towards a stable coalition structure throughout the domination relation.
  The reduced form of a game provides us with a tool for presenting a necessary and sufficient condition for general coalition formation games to guarantee that agents achieve stability on their own:
If no reduced form of the game has generalized rings, convergence to stability is guaranteed.

As an application of this result on convergence to stability, we show that games satisfying the ``common ranking property'' \citep{farrell1988partnerships} exhibit convergence to stability, but games satisfying the ``weak top coalition property'' \citep{banerjee2001core} and the ``ordinally balance property'' \citep{bogomolnaia2002stability} may lack convergence to stability.

We also analyze some economic environments in which coalitions produce an output to be divided among their members according to a pre-specified sharing rule. In such environments, the sharing rule naturally induces a game where each agent ranks the coalitions to which she belongs according to the payoffs that she gets. We focus on two types of sharing rules: Bargaining rules and rationing rules. We show that games induced by ``pairwise aligned'' bargaining rules \citep{pycia2012stability}, which include the Nash bargaining rule \citep{nash1950}, exhibit convergence to stability.  A similar result is obtained in the context of rationing for  ``parametric'' rules \citep[see][]{young1987dividing,stovall2014collective}, which include several of the most thoroughly-studied rules in the rationing literature. 

The rest of the paper is organized as follows. Section \ref{section preliminares} presents the
preliminaries. Section \ref{section reduced form} is divided into four subsections. Subsection \ref{subsection generalized ring} presents the definition of a generalized ring which is then used in Subsection \ref{subsection reduced form} to define a reduced form. Subsection \ref{subsection characterization} characterizes absorbing sets in terms of reduced forms. In Subsection \ref{aplicacion al roommate} we compare reduced forms to ``maximal stable partitions'' in roommate problems. Section \ref{aplication}  is devoted to applying our characterization result to the study of convergence to stability in several economic environments. Section \ref{Discusion} contains some final remarks. Appendix \ref{section ring en preferences} discusses the relation between two well-known concepts in the literature (rings in preferences and cycles of coalition structures) which is essential for the proof of the characterization result, presented in Appendix \ref{apendice}.

\section{Preliminaries}\label{section preliminares}

Let $N=\{1,\ldots,n\}$ be a finite set of \emph{agents}. A non-empty subset $C$ of $N$
is called a \emph{coalition}. 
Each agent $i\in N$ has a strict, transitive
\emph{preference relation} on the set of coalitions to which she belongs, denoted
by $\succ_{i}.$  Given coalitions $C$ and $C'$, when agent $i\in C\cap C^{\prime}$ prefers coalition $C$ to $C^{\prime}$ we write $C\succ_{i}C'$.  We say that $C$ \emph{is (unanimously) preferred to} $C'$, and write $C\succ C'$, if $C\succ_i C'$ for each $i\in C' \cap C$.

 A set of agents $N$ and a  preference profile for such agents $\succ_{N}=(\succ_{i})_{i\in N}$ define a
\emph{coalition formation game } which is denoted by $(N,\succ_N)$.  Let $\mathcal{K}=\{C\in 2^N :|C|\geq1 \text{ and } C\succeq_i\{i\}$ for each $i\in N\}$ be the set of \textit{permissible} coalitions of game $(N,\succ_N)$.\footnote{Here, $2^N$  denotes the collection of non-empty subsets of $N$.}  






 Let $\Pi$ denote the set of partitions of $N$ formed by permissible coalitions,  which we call \emph{coalition structures.}  A generic element of $\Pi
$ is denoted by $\pi.$ For each $\pi \in
\Pi$ and each $i \in N,$ $\pi(i)$ denotes the coalition in $\pi$ that
contains agent $i.$ 
Furthermore, given coalition structure $\pi$ and a coalition $C\in \mathcal{K}\setminus \pi$, we say that $C$ \emph{blocks} $\pi$ if $C\succ\pi(i)$ for each $i\in C$. 

The main solution concept for a coalition formation game is that of stability, namely, a coalition structure that is immune  to the deviation of
coalitions.
A coalition structure $\pi \in\Pi$ is \emph{stable} if  the existence of $C\in \mathcal{K}$ and $i\in C$ such that $C\succ_i\pi(i)$ implies the existence of $j\in C$ such that  $\pi(j)\succ_j C$. In other words, a stable coalition structure ``protects'' its coalitions from external deviations.  Hereafter, a stable coalition structure is denoted by $\pi^N.$ 

Another solution concept for a coalition formation game is that of an ``absorbing set'', which can be described as follows:
Given a dynamic process defined between coalition structures by means of a domination relation, an absorbing set is a minimal collection of coalition structures that, once entered throughout this dynamic process, is never left. 
In this paper, we adopt the
standard (myopic) dynamic process 
in which a coalition structure is replaced by a new one where a coalition of better-off agents is formed, agents that are abandoned appear as singletons,
 and all other coalitions remain unchanged.\footnote{The dynamic process just described is inspired by the $\gamma$-model by \cite{hart1983endogenous}. Other dynamic processes can be considered.  For instance, the $\delta$-model in \cite{hart1983endogenous}.
For roommate problems, \cite{tamura1993transformation} considers a similar process but abandoned agents instead of remaining single join together in a new coalition.}

Let $(N, \succ_N)$ be a coalition formation game. The \textit{domination relation} $\boldsymbol{\gg}$ over $\Pi$ is defined as follows: $\pi' \gg \pi$ if and only if there is $C \in  \mathcal{K}$ such that 
\begin{enumerate}[(i)]
\item $C \in \pi ^{\prime }$ and $C\succ \pi(i)$ for each $i\in C$,

\item for each $C'\in \pi$ such that $C' \cap C \neq \emptyset$, $%
\pi^{\prime}(j)=\{j\}$ for each $j\in C' \setminus C,$

\item for each $C' \in \pi$ such that $C' \cap C=\emptyset $, $C' \in \pi
^{\prime }$.
\end{enumerate}

Notice that $\gg$ is a binary relation that is irreflexive, antisymmetric, and not necessarily transitive.
To stress the role of coalition $C,$  $\pi'$ is said to dominate $\pi$ via $C$, and  \textit{$\pi' \gg \pi$ via $C$} is written.
Condition (i) says that each agent $i$ of coalition $C$ improves in $\pi'$ with
respect to her position in $\pi$. Condition (ii) says that the agents in coalitions from
which one or more of them depart (to form $C$) become singleton sets in $\pi'$.
Condition (iii) says that coalitions that  suffer no departures in 
$\pi $ remain unchanged in $\pi ^{\prime }$.

Given the domination relation $\gg$  between coalition structures, let $\gg^{T}$ be the \textit{transitive closure} of $\gg$. That is,  given any two coalition structures $\pi$ and $\pi'$,  $\pi^{\prime
}\gg^{T}\pi$ if and only if there is a finite sequence of coalition structures
$\pi=\pi^{0},\pi^{1},\ldots,\pi^{J}=\pi^{\prime}$ such that, for all
$j \in\{1,...,J\}$, $\pi^{j}\gg \pi^{j-1}$. Now, we are in a position to formalize the notion of absorbing set.\medskip

\begin{definition}\label{absorbing}
A non-empty set of coalition structures $\mathcal{A}\subseteq \Pi$ is an \textbf{absorbing set} whenever for each $\pi \in \mathcal{A}$ and each $\pi' \in \Pi \setminus \{\pi\},$ $$\pi' \gg^T \pi\text{ if and only if }\pi' \in \mathcal{A}.$$

\noindent If $|\mathcal{A}| \geq 3$,   $\mathcal{A}$ is said to be a \textbf{non-trivial absorbing set}. Otherwise, the absorbing set is \textbf{trivial}. 

\end{definition}\medskip

\noindent Notice that coalition structures in $\mathcal{A} $ are
symmetrically connected by the relation $\gg^{T}$, and that no coalition structure in $\mathcal{A}$ is dominated by a coalition structure outside the set. \medskip

\begin{remark}\label{remarkabsorbing}\ 
Facts on absorbing sets.
\begin{enumerate}[(i)] 
\item There is always an absorbing set.
\item An absorbing set $\mathcal{A}$ contains no stable coalition structure if and only if  $|\mathcal{A}| \geq3.$ 
\item $\pi^{N}$ is a stable coalition structure if and only if  $\{ \pi^{N}\}$ is a trivial absorbing set.
\item For  each non-stable coalition structure $\pi \in \Pi$, there are an absorbing set $\mathcal{A}$ and a coalition structure $\pi'\in \mathcal{A}$ such that $\pi' \gg^T \pi$. 
\end{enumerate}
\end{remark}\medskip

\noindent  Remark \ref{remarkabsorbing} (i) is implied by the finiteness of the set of coalition structures.
Remark \ref{remarkabsorbing} (ii)  is implied by the antisymmetry of $\gg$. Remark \ref{remarkabsorbing} (iii) says that a stable coalition structure is, by definition, immune to the deviation of coalitions and is, therefore, a maximal element for the domination relation $\gg$.  Remark \ref{remarkabsorbing} (iv) says that from any non-stable coalition structure, there is a finite sequence of such structures that reaches a coalition structure in an absorbing set. This solution concept can be illustrated with the following example.\medskip

\begin{example}\label{ejemplo ring}

Consider the game given by the following table: \medskip
{\small
\begin{center}

\begin{tabular}{@{}ccccccccc@{}}
\hline
$\boldsymbol{1}$ && $\boldsymbol{2}$ && $\boldsymbol{3}$ && $\boldsymbol{4}$ && $\boldsymbol{5}$   \\ \hline
$12$ && $23$ && $34$ && $45$ && $15$ \\ 
$123$ && $123$ && $123$ && $34$ && $45$\\ 
$15$ && $12$ && $23$ && $4$  && $5$ \\ 
$1$ && $2$ && $3$ && &&   \\  \hline
\end{tabular} 
\end{center}
}\medskip
\begin{figure}[h!]
    \centering
    \begin{tikzpicture}[scale=0.30]
\node (1) at (6,-8.5) {{\small \textcolor{blue}{45}}};
\node (2) at (-2,-3) {{\small \textcolor{blue}{23}}};
\node (3) at (-0.5,4.5) {{\small \textcolor{blue}{15}}};
\node (4) at (12.5,4.5) {{\small \textcolor{blue}{34}}};
\node (5) at (14,-3) {{\small \textcolor{blue}{12}}};
\node (7) at (6,6) {{\small\{15,23,4\}}};
\node (8) at (13,1) {{\small\{15,34,2\}}};
\node (9) at (10.5,-6) {{\small\{12,34,5\}}};
\node (10) at (1.5,-6) {{\small\{12,45,3\}}};
\node (11) at (-1,1) {{\small\{23,45,1\}}};

\draw[blue,-{Latex[length=1.3mm]}] (7) to [bend left] (8);
\draw[blue,-{Latex[length=1.3mm]}] (8) to [bend left] (9);
\draw[blue,-{Latex[length=1.3mm]}] (9) to [bend left] (10);
\draw[blue,-{Latex[length=1.3mm]}] (10) to [bend left] (11);
\draw[blue,-{Latex[length=1.3mm]}] (11) to [bend left] (7);

\end{tikzpicture}
    \caption{Non-trivial absorbing set of Example \ref{ejemplo ring}.}
    \label{figura ejemplo 1}
\end{figure}
This game has two absorbing sets: a trivial absorbing set corresponding to the unique stable coalition structure $\{123,45\}$ and a non-trivial absorbing set depicted in Figure \ref{figura ejemplo 1}. 
\hfill $\Diamond$
 \end{example}\medskip

\section{A characterization of absorbing sets}\label{section reduced form}

In this section, we present a characterization of an absorbing set throughout a collection of sets of coalitions that condenses all its relevant information,  which we call  ``reduced form''.
 Before we state the characterization result,  we present the key ingredients of that construct.

 \subsection{Generalized rings}\label{subsection generalized ring}

As argued in the Introduction, there are sets of coalitions presenting cyclical behavior, and some of these sets induce the cyclical behavior of coalition structures in a non-trivial absorbing set. In this subsection, we formalize these sets that we call generalized rings.


For the game in Example \ref{ejemplo ring}, consider the collection of coalitions $\mathcal{B}=\{12,15,45,34,23\}.$  This collection presents cyclical behavior among coalitions since $12\succ 15\succ 45 \succ 34 \succ 23 \succ 12$. Furthermore, it induces cyclical behavior among coalition structures in the unique non-trivial absorbing set, because for each coalition in $\mathcal{B}$ there are two coalition structures of the non-trivial absorbing set such that one dominates the other one via that coalition, for instance,  $\{12,3,45\}$ dominates $\{12,34,5\}$ via $45$. Note that collection $\{123,15,45,34\}$ also presents cyclical behavior among coalitions because  $123\succ 15 \succ 45\succ 34 \succ 123$. However, it does not induce cyclical behavior among coalition structures since there is no pair of coalition structures in the non-trivial absorbing set such that one dominates the other via $123.$ 

Notice that the previous example suggests that we need to request an additional feature to a collection of coalitions presenting cyclical behavior among themselves in order to induce a cyclical behavior among coalition structures of a non-trivial absorbing set. A generalized ring is a collection of coalitions satisfying such an additional feature. Before presenting it formally, we define some notions that will be useful later.

Let $\mathcal{B} \subseteq \mathcal{K}$ be a collection of non-singleton coalitions. 
A set $\mathcal{M} \subseteq \mathcal{B}$ is \textit{maximal for} $\mathcal{B}$ if:
\begin{enumerate}[(i)]
\item $C, C' \in \mathcal{M}$ implies $C \cap C'=\emptyset$, and 
\item for each $C \in \mathcal{B} \setminus \mathcal{M}$  there is $C' \in  \mathcal{M}$ such that $C \cap C' \neq \emptyset.$
\end{enumerate}

Condition (i) says that coalitions of a maximal set are pairwise disjoint. Condition (ii) establishes that the set cannot be enlarged.
Denote by $\mathcalligra{M}\, _{\mathcal{B}}$ the collection of all maximal sets for $\mathcal{B}.$

 Given a collection of non-singleton coalitions  $\mathcal{B} \subseteq \mathcal{K}$ and a coalition $C \in \mathcal{K}\setminus \mathcal{B},$  we say that  $C$ \textit{breaks} $\mathcal{B}$  if  there is a maximal set $\mathcal{M}\in \mathcalligra{M}\, _{\mathcal{B}}$ such that:

\begin{enumerate}[(i)]
\item  there is a coalition $C'\in \mathcal{M}$ such that $C \cap C' \neq \emptyset,$ and 

\item $C \succ C''$ for each $C'' \in \mathcal{M}$ such that $C'' \cap C\neq \emptyset$. 
\end{enumerate}
Condition (i) says that there is always a coalition in  $\mathcal{M}$ with agents in common with $C$. Condition (ii) says that coalition $C$ is preferred to each coalition in $\mathcal{M}$ that has agents in common with $C$.




\medskip

 \begin{definition}\label{def generalize ring}
 Given a game $(N,\succ_N)$,  a set of non-singleton coalitions  $\mathcal{B} \subseteq \mathcal{K}$ is a \textbf{generalized ring} if 
\begin{enumerate}[(i)]
\item  $C \succ^T C'$ for each pair $C, C' \in \mathcal{B}$ with $C\neq C'$,\footnote{Here $\succ^T$ denotes the transitive closure of relation $\succ$. That is,  $C'\succ^{T}C$ if and only if there is a finite sequence of coalitions $C=C_{0},C_{1},\ldots,C_{J}=C'$ such that, for all
$j \in\{1,\ldots,J\}$, $C_{j}\succ C_{j-1}$.}
\item for each maximal set $\mathcal{M}\in \mathcalligra{M}\,_\mathcal{B}$ \  there is $C \in \mathcal{B} \setminus \mathcal{M}$ such that $C$ breaks $\mathcal{M}.$
\end{enumerate}     
\end{definition}\medskip

Notice that $|\mathcal{B}|\geq 3$.
Condition (i)  says that each coalition in the set is transitively preferred
to any other coalition in the set. Condition (ii) says that for each set of disjoint coalitions of the generalized ring, there is another coalition of the generalized ring that breaks the set.

Consider the set  $\{15,123,34,45\}$ in Example \ref{ejemplo ring}. This set fulfills Condition (i) but not Condition (ii) in Definition \ref{def generalize ring} because, for example, the maximal set $\{15,34\}$ cannot be broken by any other coalition \emph{in the set}. 
 Now, consider the set  $\{15,12,23,34,45\}$ in Example \ref{ejemplo ring}. This set fulfills both Conditions (i) and (ii) in Definition \ref{def generalize ring}. It is easy to see that such a set consists of all the non-singleton coalitions that belong to a coalition structure of the non-trivial absorbing set depicted in Figure \ref{figura ejemplo 1}.

The following example illustrates a set of coalitions fulfilling Condition (ii)  but not Condition (i) in Definition \ref{def generalize ring}.\medskip

\begin{example}\label{ejemplo 2 rings}

Consider the game given by the following table: \medskip
{\small
\begin{center}
\begin{tabular}{@{}ccccccccccccccc@{}}
\hline
$\boldsymbol{1}$ && $\boldsymbol{2}$ && $\boldsymbol{3}$ && $\boldsymbol{4}$ && $\boldsymbol{5}$ && $\boldsymbol{6}$  \\ \hline
$12$ && $23$ && $13$ && $45$ && $56$ && $46$ \\ 
$123$ &&$123$&&$123$&&$456$&&$456$&&$456$\\
$13$ && $12$ && $23$ && $46$ && $45$ && $56$ \\ 
$1$ && $2$ && $3$ && $4$ && $5$ && $6$   \\  \hline
\end{tabular} 
\end{center}
}\medskip

\noindent In this game there are two generalized rings: $\{13,12,23\}$ and $\{46,45,56\}$. Notice that their union fulfills Condition (ii) but not Condition (i). Therefore, such a union is not a generalized ring.
\hfill $\Diamond$
 \end{example}\medskip

By delving deeper into the analysis we can distinguish two different types of generalized rings depending on whether each coalition structure of a non-trivial absorbing set includes a maximal set of the generalized ring or not. In Example \ref{ejemplo ring}, the first situation occurs. The maximal sets of the generalized ring $\{12,23,34,45,15\}$  are $\{12,34\},\{12,45\},\{23,45\},$ $\{23,15\}$, and $\{34,15\}.$ Each coalition structure of the non-trivial absorbing set includes one of these maximal sets.
This is not always the case, though. To show this, we present the following example.\medskip

\begin{example}\label{ejemplo tipo 2}
Consider the game given by the following table: 
\medskip

{\small
\begin{center}
\begin{tabular}{@{}ccccccccccccccc@{}}
\hline
$\boldsymbol{1}$ && $\boldsymbol{2}$ && $\boldsymbol{3}$ && $\boldsymbol{4}$ && $\boldsymbol{5}$ && $\boldsymbol{6}$ && $\boldsymbol{7}$ && $\boldsymbol{8}$\\ \hline
$12$ && $23$ && $356$ && $145$ && $356$ && $678$&& $78$ && $678$\\ 
$145$ && $12$ && $23$ && $46$ && $145$ && $46$&& $678$&& $78$ \\ 
$1$ && $2$ && $3$ && $4$  && $5$ && $356$&& $7$ && $8$\\ 
 &&  &&  && && &&$6$ && &&   \\  \hline
\end{tabular} 
\end{center}
}
\begin{figure}[h!]
    \centering
    \begin{tikzpicture}[scale=0.37]
\node (1) at (6,-3.5) {{\small \textcolor{blue}{12}}};
\node (2) at (-2.7,2) {{\small \textcolor{blue}{23}}};
\node (3) at (0,13) {{\small \textcolor{blue}{356}}};
\node (4) at (12,13) {{\small \textcolor{blue}{46}}};
\node (5) at (14,0) {{\small \textcolor{blue}{145}}};

\node (16) at (16,15) {{\small \textcolor{blue}{12}}};
\node (17) at (22,1) {{\small \textcolor{blue}{46}}};
\node (18) at (17.7,-3) {{\small \textcolor{blue}{12}}};
\node (20) at (10,-10.5) {{\small \textcolor{blue}{46}}};
\node (21) at (2,-0.3) {{\small \textcolor{blue}{23}}};
\node (22) at (-5.8,-3) {{\small \textcolor{blue}{46}}};
\node (23) at (-7.3,5.7) {{\small \textcolor{blue}{145}}};
\node (24) at (3.5,4.5) {{\small \textcolor{blue}{23}}};
\node (25) at (-9.9,1) {{\small \textcolor{blue}{145}}};
\node (26) at (-4,15) {{\small \textcolor{blue}{356}}};
\node (27) at (6,-13) {{\small \textcolor{blue}{23}}};

\node (7) at (6,19) {{\small\{1,2,356,4,78\}}};
\node (8) at (15,5) {{\small\{1,2,3,46,5,78\}}};
\node (9) at (10.5,-6) {{\small\{145,2,3,6,78\}}};
\node (10) at (1.5,-6) {{\small\{12,3,4,5,6,78\}}};
\node (11) at (-3,5) {{\small\{1,23,4,5,6,78\}}};
\node (12) at (22,8) {{\small\{12,356,4,78\}}};
\node (13) at (17,-10) {{\small\{12,3,46,5,78\}}};
\node (14) at (-5,-10) {{\small\{1,23,46,5,78\}}};
\node (15) at (-10,8) {{\small\{145,23,6,78\}}};

\draw[blue,-{Latex[length=1.3mm]}] (7) to [bend left] (8);
\draw[blue,-{Latex[length=1.3mm]}] (8) to [bend left] (9);
\draw[blue,-{Latex[length=1.3mm]}] (9) to [bend right] (10);
\draw[blue,-{Latex[length=1.3mm]}] (10) to [bend left] (11);
\draw[blue,-{Latex[length=1.3mm]}] (11) to [bend left] (7);
\draw[blue,-{Latex[length=1.3mm]}] (7) to [bend left] (12);
\draw[blue,-{Latex[length=1.3mm]}] (12) to [bend left] (13);
\draw[blue,-{Latex[length=1.3mm]}] (8) to [bend left] (13);
\draw[blue,-{Latex[length=1.3mm]}] (8) to [bend right] (14);
\draw[blue,-{Latex[length=1.3mm]}] (9) to [bend right] (15);
\draw[blue,-{Latex[length=1.3mm]}] (10) to [bend right] (13);
\draw[blue,-{Latex[length=1.3mm]}] (13) to [bend left] (14);
\draw[blue,-{Latex[length=1.3mm]}] (11) to [bend right] (14);
\draw[blue,-{Latex[length=1.3mm]}] (11) to [bend left] (15);
\draw[blue,-{Latex[length=1.3mm]}] (14) to [bend left] (15);
\draw[blue,-{Latex[length=1.3mm]}] (15) to [bend left] (7);

\end{tikzpicture}
    \caption{The absorbing set of Example \ref{ejemplo tipo 2}.}
    \label{figura ejemplo 3}
\end{figure}
\noindent The unique generalized ring is $\{145,12,23,356,46\}$ and its  maximal sets are $\{145,23\},$ $\{12,356\},$ $\{12,46\},$ and $\{23,46\}$.  
Notice that coalition $356$ breaks the maximal set $\{145,23\}$ and has a non-empty intersection with \emph{two} coalitions: $145$ and $23$.
This means that coalition structures $\{1,2,356,4,78\}$, $\{1,2,3,46,5,78\}$, $\{145,2,3,6,78\}$, $\{12,3,4,5,6,78\}$, and $\{1,23,4,5,6,78\}$ belong to the non-trivial absorbing set depicted in Figure \ref{figura ejemplo 3}. Note that these coalition structures do not include any maximal set of the generalized ring. 
\hfill $\Diamond$
\end{example}

We now formally define the two types of generalized rings.\medskip

\begin{definition}
  A generalized ring $\mathcal{B}$ is \textbf{compact} if for each  $ \mathcal{M}\in \mathcalligra{M}\, _ \mathcal{B}$ \  and each $C\in \mathcal{B}$ such that $C$ breaks $ \mathcal{M}$, there is a unique $C'\in \mathcal{M}$ such that $C\cap C'\neq \emptyset.$ Otherwise, we say that the generalized ring is \textbf{non-compact}.
     
\end{definition}\medskip

Note that in a non-compact generalized ring there is  $ \mathcal{M}\in \mathcalligra{M}\, _ \mathcal{B}$, a coalition $C\in \mathcal{B}$ such that $C $ breaks $\mathcal{M}$, and at least two coalitions in  $\mathcal{M}$ that intersect $C$.
Observe that the generalized ring $\{12,23,34,45,15\}$ in Example \ref{ejemplo ring} is compact. In contrast, the generalized ring $\{145,12,23,356,46\}$ in Example \ref{ejemplo tipo 2} is non-compact.

\subsection{The reduced form of a coalition formation game}\label{subsection reduced form}

We have shown a close connection between non-trivial absorbing sets and generalized rings. One or more generalized rings generate a non-trivial absorbing set. However, not all generalized rings can generate a non-trivial absorbing set. The concept of reduced form of a coalition formation game enables us to distinguish between those generalized rings that generate a non-trivial absorbing set and those that do not. The ingredients of a reduced form are generalized rings,  sets consisting of one non-singleton coalition (called ``fixed components''), and a set of singletons.

Consider again the game in Table \ref{table para la reduced form} presented in the Introduction. Recall that this game has a unique non-trivial absorbing set. Furthermore, there are two generalized rings, $\{13,12,23\}$ and $\{45,46,56\}$, and:
\begin{itemize}
    \item Each coalition structure of the absorbing set contains coalition $78$.
    \item There are coalition structures of the absorbing set that contain no coalition of the generalized ring $\{13,12,23\}$, for instance, coalition structure $\{1,2,3,45,6,78\}$.
    \item Each coalition structure of the absorbing set contains a coalition of the generalized ring $\{45,46,56\}$. 

\end{itemize}

Notice that, from any coalition structure of the absorbing set, agent $6$ has the incentive to deviate to external coalition $678$. As coalition $78$ belongs to each coalition structure of the absorbing set, $78$ ``impedes coalition $678$ from being formed''. Generalized rings can also play this role. The formal definition of this notion is now presented. Given a set of coalitions $\mathcal{D}$, denote by $N(\mathcal{D})$  the set of agents that belong to (at least) one coalition in $\mathcal{D},$ that is, $N(\mathcal{D}) \equiv \bigcup_{C \in \mathcal{D}}C.$ 

\begin{definition}
    Given a generalized ring or a fixed component $\mathcal{B}$ and   a coalition $C \in \mathcal{K} \setminus \mathcal{B}$ such that $N(\mathcal{B}) \cap C \neq \emptyset,$ we say that \textbf{$\boldsymbol{\mathcal{B}$   impedes coalition $C}$ to be formed} if:
\begin{enumerate}[(i)]
\item $\mathcal{B}$ is a compact generalized ring and for each $\mathcal{M}\in \mathcalligra{M}\, _\mathcal{B}$ there is $C'\in\mathcal{B}$ with $C'\cap C \neq \emptyset$ and an agent $i \in C'\cap C$ such that $C' \succ_i C$.
\item $\mathcal{B}$ is a non-compact generalized ring or a fixed component and for each $C'\in \mathcal{B}$  there is an agent $i \in C'\cap C$ such that $C' \succ_i C$.\footnote{Notice that if $\mathcal{B}$ is a fixed component, there is only one such $C'$.} 
\end{enumerate}
\end{definition}\medskip

\noindent Condition (i) states that each maximal set of a compact generalized ring contains a coalition with an agent that prefers that coalition to the external coalition. Condition (ii) states that for each coalition of a non-compact generalized ring or for the coalition of each fixed component, an agent prefers that coalition(s) to the external coalition. 
Notice that, for compact generalized rings, maximal sets are the objects that impede the formation of external coalitions because only the maximal sets of a compact generalized ring appear in all the coalition structures of an associated absorbing set (see Figure \ref{figura ejemplo 1}). For non-compact generalized rings, however, their coalition(s) are the objects that impede the formation of external coalitions because there are coalition structures of an associated absorbing set that only have one coalition of the non-compact generalized ring.
For fixed components, the situation is similar to non-compact generalized rings because they consist of only one coalition.

Our proposal of reduced form for the non-trivial absorbing set of the game presented in Table \ref{table para la reduced form} contains the generalized ring $\{45,46,56\}$, the fixed component $\{78\}$, and the set of singletons $\{1,2,3\}$.
Notice that the generalized ring $\{45,46,56\}$ and the fixed component  $\{78\}$ of this reduced form ``protect'' each other from external blocking coalitions. However, they do not ``protect'' the generalized ring $\{13,12,23\}$, and for that reason, agents $1,2,3$ belong to the set of singletons of the reduced form. Next, we formalize the notion of protection. A collection (of sets of coalitions) $\mathcalligra{R} \  \ $ is \emph{complete} if it consists of generalized rings, fixed components, or a set of singletons that fulfills  (i) $\bigcup_{\mathcal{D}\in \mathcalligra{R} \  \ }N(\mathcal{D})=N$, together with (ii) $N(\mathcal{D})\cap N(\mathcal{D}')=\emptyset$ for each pair $\mathcal{D},\mathcal{D}'\in \mathcalligra{R} \  .$

\medskip
 \begin{definition}
    Let $\mathcalligra{R} \  \ $ be a complete collection.  A generalized ring or a fixed component $\mathcal{B}$ is said to be  \textbf{protected by $\mathcalligra{R} \ \ $} if for each coalition $C$ that breaks $\mathcal{B}$ there is $\mathcal{B}' \in \mathcalligra{R} \ \ $ such that  $\mathcal{B}'$ impedes  $C$ to be formed.\footnote{Notice that  $\mathcal{B}$ need not be included in $\mathcalligra{R} \ $.} 
 \end{definition}\medskip

Now, we are in a position to formally define the notion of a reduced form of a game.
\medskip

\begin{definition}\label{reduced form} Let $(N,\succ_N)$ be a coalition formation game. A \textbf{reduced form}  of $(N,\succ_N)$ is a complete collection $\mathcalligra{R} \ \ $  that may include  generalized rings, fixed components, or a set of singletons (denoted by $\mathcal{S}$) and satisfies the following:

\begin{enumerate}[(i)]

\item each generalized ring and each fixed component in $\mathcalligra{R} \ \ $  is protected by  $\mathcalligra{R} \ .$
\item each generalized ring and each fixed component not in $\mathcalligra{R} \ \ $ generated by agents in $\mathcal{S}$  is \emph{not} protected by $\mathcalligra{R} \ .$
\end{enumerate}
 
\end{definition}
\medskip

\noindent Condition (i) states that the reduced form protects itself from external deviations. Condition (ii) is a maximality condition: it states that if a generalized ring or a fixed component is added to the collection, the resulting collection no longer protects itself from external deviations. 

In Example \ref{ejemplo ring}, the game has two reduced forms: $\{\{15,12,23,34,45\}\}$ and  $\{\{123\},$ $\{45\}\}.$ The first reduced form contains only one generalized ring, while the second contains only two fixed components. In Example \ref{ejemplo 2 rings}, the game has four reduced forms: $\{\{123\},\{456\}\}$, $\{\{13,12,23\},\{456\}\}$, $\{\{123\},\{46,45,56\}\}$, and $\{\{13,12,23\},\{46,45,56\}\}$. The game presented in Table \ref{table para la reduced form} in the Introduction has only one reduced form, consisting of the generalized ring $\{45,46,56\}$, the fixed component $\{78\}$, and the set of singletons $\{1,2,3\}$.

\subsection{The characterization result}\label{subsection characterization}

In this subsection, we relate reduced forms of a game with absorbing sets of such a game. 
Recall that, by Remark \ref{remarkabsorbing}, a trivial absorbing set consists of a collection that has as its unique element a stable coalition structure. The following remark shows how a trivial absorbing set can be identified with a reduced form that has only fixed components and (possibly) a set of singletons.\medskip

\begin{remark}\label{remark reduced form con stable partition}
Suppose there is a reduced form that has no generalized ring. In that case, the coalition formation game has a trivial absorbing set, and by Remark \ref{remarkabsorbing} (iii) it has a stable coalition structure.
\end{remark}\medskip

Even though a reduced form and an absorbing set are different concepts, we show they are related. We characterize absorbing sets in terms of reduced forms. Formally,

\begin{theorem}\label{bijection}
For each coalition formation game,  each absorbing set generates a reduced form and, conversely, each reduced form generates an absorbing set.
\end{theorem}
The proof of Theorem \ref{bijection} is relegated to Appendix \ref{apendice}. Essential to this proof is the relation between rings in preferences on one hand and cycles of coalition structures on the other. This relation is analyzed in Appendix \ref{section ring en preferences}.

The game presented in Table \ref{table para la reduced form} in the Introduction has only one absorbing set, with fourteen coalition structures. Such absorbing set is associated to reduced form $\{\{1,2,3\},\{45,46,56\},\{78\}\}.$ The game presented in Example \ref{ejemplo ring} has two absorbing sets, one trivial and one non-trivial with five coalition structures. They are associated with reduced forms $\{\{123\},\{45\}\}$ and $\{\{15,12,23,34,45\}\}$, respectively. The game presented in Example \ref{ejemplo 2 rings} has four absorbing sets: one trivial, associated with reduced form $\{\{123\},\{456\}\}$; two non-trivial with three coalition structures each, associated with reduced forms   $\{\{13,12,23\},\{456\}\}$ and  $\{\{123\},\{46,45,56\}\}$, respectively; and one non-trivial with nine coalition structures, associated with reduced form $\{\{13,12,23\},\{46,45,56\}\}$. Notice that no intrinsic property of a generalized ring generates a non-trivial absorbing set. What is important is whether a generalized ring belongs to a reduced form. But, in a game, a generalized ring can belong to some reduced forms (and not to others), or even to none. For instance, in the example of Table \ref{table para la reduced form}, the generalized ring $\{13,12,23\}$ does not belong to the unique reduced form; whereas in Example \ref{ejemplo 2 rings} the generalized ring $\{46,45,56\}$ belongs to only two of the four different reduced forms.

\subsection{Reduced forms in roommate problems}\label{aplicacion al roommate}
An important class of coalition formation games is the class of roommate problems, introduced by \cite{gale1962college}. In these games, each agent
has preferences over all coalitions of cardinality two to which she belongs.
As is known, a roommate problem may not admit stable coalition structures. In this section we show that the notion of reduced form has already been identified in roommate problems under the name of ``maximal stable partitions'' by \cite{inarra2013absorbing}. 
To do this, we compare the notion of reduced form with the notion of ``stable partition'' \citep{tan1991necessary} in roommate problems. We show that every reduced form is a stable partition, but only ``maximal stable partitions'' are reduced forms. 

 \cite{tan1991necessary}
proves that a roommate problem has no stable coalition structures if and only if there
is a ``stable partition'' with an odd ring.\footnote{For a formal definition of ``ring'' see Appendix \ref{section ring en preferences}.}
In our terminology, a complete collection $\mathcalligra{P}\ \ $ is a \emph{stable partition} if, whenever coalition $C$ breaks a fixed component or a ring $\mathcal{B}\in\!\! \mathcalligra{P}\ \ $ there is another fixed component or ring $\mathcal{B}'\in\!\! \mathcalligra{P}\ \ $  such that $C\setminus N(\mathcal{B})\subsetneq N(\mathcal{B}')$ and $R \succ C$ for each $R\in \mathcal{B}'$ with $R\cap C \neq\emptyset$. Thus, the notion of protection in \cite{tan1991necessary} is weaker than ours.  This can be illustrated  with the following example: 


\begin{example}[Example 2 in \citealp{inarra2013absorbing}]\label{ejemploMolis} 
Consider the  game  given by this table:
{\small
\[
\ra{1.1}
\begin{tabular}{@{}ccccccccccccccccccccc@{}}\toprule

$\mathbf{1}$ && $\mathbf{2}$ && $\mathbf{3}$ && $\mathbf{4}$ && $\mathbf{5}$ && $\mathbf{6}$ && $\mathbf{7}$ && $\mathbf{8}$ && $\mathbf{9}$ && \textbf{a} \\ 
\cmidrule{1-1} \cmidrule{3-3} \cmidrule{5-5} \cmidrule{7-7} \cmidrule{9-9} \cmidrule{11-11} \cmidrule{13-13} \cmidrule{15-15} \cmidrule{17-17} \cmidrule{19-19} 
$12$ && $23$ && $13$ && $47$ && $58$ && $69$ && $57$ && $68$ && $49$ && $a$  \\ 
$13$ &&  $12$ && $23$ && $48$ && $59$ && $67$ && $67$ && $48$  && $59$ &&   \\ 
$14$ && $24$ && $34$ && $49$ && $57$ && $68$ && $17$ && $58$ &&  $69$ &&   \\ 
$15$ && $25$ & & $35$ && $45$ && $45$ && $46$ && $47$ && $78$ && $79$ &&  \\ 
$16$ && $26$ && $36$ &&$46$ && $56$ && $6$ &&$ 79$ && $89$ && $89$ &&  \\ 
$17$ && $27$ && $37$ && $14$ && $5$ &&  && $78$ && $8$ && $9$ && \\ 
$18$ && $28$ && $38$ && $24$ &&  &&  && $7$ &&  &&  && \\ 
$19$ && $29$ && $39$ && $34$ &&  &&  &&  &&  &&  && \\
$1$ && $2$ && $3$ && $4$ && &&  &&  &&  &&  && \\

\bottomrule
\end{tabular}%
\]}

\noindent This example has three stable partitions: $\{\{12,23,13\},\{48\},\{59\},\{67\},%
\{a\}\}$,\\ $\{\{12,23,13\},$ $\{49\},\{57\},\{68\},\{a\}\},$ and $\{\{12,23,13\},\{47
\},\{58\},\{69\},\{a\}\}$.\\ The first two are also reduced forms.  The third one is a complete collection but not a reduced form  because fixed component $\{47\}$ is not protected by the collection: Coalition $17$ breaks fixed component
$\{47\}$ and the collection does not impedes the formation of coalition $17$. Thus,  for the roommate problem, the notion of stable partition is weaker than the notion of reduced form. \hfill $\Diamond$
\end{example}

Following \cite{inarra2013absorbing}, we use the term \textit{maximal stable
partition} to refer to those stable partitions with the maximal set of satisfied agents,
i.e. agents with no incentive to change partners. The following result can then be established:

\begin{proposition}\label{propisition3}
For each roommate problem, each maximal stable partition generates a reduced form and, conversely, each reduced form generates a maximal stable partition.
\end{proposition}
\begin{proof}
Theorem 1 in \cite{inarra2013absorbing} proves that there is a bijection between maximal stable partitions and absorbing sets. Our Theorem \ref{bijection} states that there is a bijection between absorbing sets and reduced forms. Therefore, the result follows straightforwardly.
\end{proof}
\section{Application: Convergence to stability}\label{aplication}

In this section, we study an application of our characterization result.  
 A relevant question for coalition formation games with a stable coalition structure (henceforth, \emph{stable} coalition formation games) is whether agents always reach stability when left to interact on their own, or whether there is a need to introduce an arbitrator so that agents can achieve stability. There are coalition formation games, such as two-sided matching models (one-to-one and many-to-one) and one-sided matching models (roommate problems), in which no arbitrator is needed to reach stability \citep[see][for more details]{roth1990random,chung2000existence,klaus2005stable,kojima2008random,eriksson2008instability,diamantoudi2004random}.\footnote{For $TU$-games, the similar question of accessibility to the core is studied in \cite{koczy2004coalition}.} However, this is not always the case in more general coalition formation games. 
 
 The reduced form of a game provides us with a tool for presenting a necessary and sufficient condition for general coalition formation games to guarantee that agents achieve stability on their own. 
We say that a stable coalition formation game $(N,\succ _{N})$ exhibits \textit{convergence to stability} if for each non-stable coalition structure $\pi\in \Pi $ there is a stable coalition structure $\pi ^{N}\in \Pi $ such that $\pi ^{N}\gg ^{T}\pi $. 
Given a stable coalition formation game, it is clear that if no reduced form has a generalized ring, then the game exhibits convergence to stability. This is because each absorbing set then is trivial. However, if there is a reduced form with a generalized ring, then the game does not exhibit convergence to stability. Formally,

\begin{proposition}\label{proposition5}
A stable coalition formation game exhibits convergence to stability if and only if
none of its reduced forms has a generalized ring.
\end{proposition}
\begin{proof} 
Let $(N,\succ _{N})$ be a stable coalition formation game.

\noindent $(\Longrightarrow )$  Assume that $(N,\succ _{N})$ has a reduced form
with a generalized ring. By Proposition \ref{absorbing and stable decomp} in Appendix \ref{apendice}, the reduced form induces a non-trivial absorbing set $\mathcal{A
}$. Let $\pi ^{N}$ be a stable coalition structure, so $\pi^N \in \Pi \setminus \mathcal{A}$. Thus, by
Definition \ref{absorbing} there is no  $\pi \in \mathcal{A%
}$  such that $\pi^N \gg ^{T}\pi $. Therefore, $(N,\succ _{N})$ does not exhibit convergence to stability.

\noindent$(\Longleftarrow )$ Assume that $(N,\succ _{N})$ has every reduced form
with no generalized rings. By Remark \ref{remark reduced form con stable partition}, this means that each reduced form can be identified with a stable coalition structure, so the game has only trivial absorbing sets. By Remark \ref{remarkabsorbing} (iii) and (iv), for each non-stable coalition structure $\pi \in \Pi $ there is a stable coalition
structure $\pi ^{N}\in \Pi $ such that $\pi ^{N}\gg ^{T}\pi $. \
\end{proof}

Recall that, for the particular case of a roommate problem, there are either trivial or non-trivial absorbing sets \citep[see][for more details]{diamantoudi2004random,inarra2013absorbing}. Thus, if a roommate problem has a reduced form with no generalized rings, then \emph{all} reduced forms of the problem have no generalized rings. Therefore, the following result holds:
\begin{corollary}\label{convergencia en roommate}
    A stable roommate problem always exhibits convergence to stability.
\end{corollary}

Several domain restrictions have been studied throughout the literature to guarantee stability for general coalition formation games. Some of them are the ``common ranking property'' \citep{farrell1988partnerships}, the ``weak top coalition property'' \citep{banerjee2001core}, the ``ordinally balance property'' \citep{bogomolnaia2002stability}, and the ``pairwise alignment property'' \citep{pycia2012stability}.

A coalition formation game satisfies the \emph{common ranking property} if there is an ordering $\vartriangleright$ over $\mathcal{K}$ such that, for each $i\in N$ and each $C,C' \in \mathcal{K}$, we have $C \succ_i C'$ if and only if  $C\vartriangleright C'.$ It is straightforward to see that a game satisfying this property has no generalized ring, so the following result holds.
\begin{corollary}
    A coalition formation game satisfying the common ranking property exhibits convergence to stability.
\end{corollary}

The following example presents two games, one satisfying the ordinally balanced property\footnote{We say that a collection of coalitions $\mathcal{B}\subseteq N$ is \emph{balanced} if there is a vector of positive weights $\lambda _{S}$, such
that for each agent $i\in N,$ $\sum_{S \in \mathcal{B}: i \in S}\lambda_S=1$ \citep[see][]{bondareva1963some, shapley1967balanced}. 
 A coalition formation game satisfies the \emph{ordinally balance property} if for each balanced
collection of coalitions $\mathcal{B}$ there is a coalition structure $\pi\in \Pi $ such
that for each $i\in N$ there is $S\in \mathcal{B}$ with $i\in S$ such that $\pi(i)\succ _{i}S$.} and the other satisfying the weak top coalition property\footnote{A coalition $W\subseteq G\subseteq N$, is a \emph{weak top coalition} of $G$ if
it has an ordered coalition structure $(S_{1},...,S_{\ell})$ such that (i) any agent in $S_{1}$ weakly prefers $W$ to any subset of $G$ and (ii) for any $k>1$, any agent in $S_{k}$ needs cooperation of at least one agent in $\cup_{m<k}S_{m} $ in order to form a strictly better coalition than $W$ \citep{banerjee2001core}. A game satisfies the \emph{weak top coalition property} if each coalition has a weak top coalition.}. Although these classes of games impose some degree of commonality on agents' preferences guaranteeing stability, they may lack convergence to stability.

\begin{example}{\citep[see][Section 4]{bogomolnaia2002stability}.}\label{ejemplo 4}
Consider the following two coalition formation games: 
\begin{equation*}
\small{
\ra{1.1}
\begin{tabular}{@{}ccccccccccccc@{}}\toprule
$\boldsymbol{1}$ &  & $\boldsymbol{2}$ && $\boldsymbol{3}$  \\
\cmidrule{1-1} \cmidrule{3-3} \cmidrule{5-5} 
$12$ && $23$ && $13$ \\ 
$\mathbf{123}$ && $\mathbf{123}$ && $\mathbf{123}$ \\ 
$13$ && $12$ && $23$ \\ 
$1$ && $2$ && $3$ \\ \bottomrule
\end{tabular}%
\hspace{60 pt}
\ra{1.1}
\begin{tabular}{@{}ccccccccccccc@{}}\toprule
$\boldsymbol{1}$ &  & $\boldsymbol{2}$ && $\boldsymbol{3}$  \\
\cmidrule{1-1} \cmidrule{3-3} \cmidrule{5-5} 
$\mathbf{123}$ && $23$ && $13$ \\ 
$12$ && $12$ && $\mathbf{123}$ \\ 
$13$ && $\mathbf{123}$ && $23$ \\ 
$1$ && $2$ && $3$ \\ \bottomrule
\end{tabular}
}
\end{equation*}%
The game in the first table is ordinally balanced and the one in the second
table satisfies the weak top coalition property. In both cases, there are two reduced forms: $\{\{123\}\}$  with no generalized ring, guaranteeing stability; and $\{\{13,12,23\}\}$ with a generalized ring, implying lack of convergence by Proposition \ref{proposition5}.
\hfill $\Diamond$
\end{example}

In the remaining of this section, we apply Proposition \ref{proposition5} to analyze convergence to stability in games induced by bargaining solutions and rationing rules.

\subsection{Coalition formation games and bargaining solutions}\label{bargaining}

\cite{pycia2012stability} presents a model in which there is a set of agents, each endowed with a utility function, who form coalitions that produce outputs to be distributed among its members. He shows that under a rich domain of preferences, and some restrictions on coalitions, there is a stable coalition structure for each preference profile if and only if agents' preferences satisfy pairwise alignment. Agents' preferences are pairwise aligned if any two agents rank coalitions that contain both of them in the same way. Formally, a game $(N,\succ_N)$ satisfies the \emph{pairwise aligned property} if for all $C,C' \in \mathcal{K}$ and all $i,j\in
C\cap C^{\prime }$ it holds that $C\succ _{i}C^{\prime }$ if and only if $C\succ _{j}C^{\prime }.$

Given a set of agents $N$ and a set of coalitions $\mathcal{K} \subseteq 2^N \setminus \{\emptyset\},$ a  \textit{coalitional bargaining  problem} is a tuple $(U_{N},y(C)_{C\in \mathcal{K}})$ where $U_{N}=(U_{i})_{i\in N}$ is a  vector of utility functions $U_{i}:\mathbb{R}_+ \longrightarrow  \mathbb{R}_+$ and, for each $C\in \mathcal{K},$ $y(C)$ is the output produced by coalition $C$. When agent $i\in C$ gets the share $x$
of output $y(C)$ her utility gives her $U_{i}(x).$ Given  $C\in \mathcal{K}$, the \textit{bargaining problem for $C$} is $(U_C, y(C))$ where $U_C=(U_i)_{i\in C}$ is the utility vector of agents in $C$ and $y(C)$ is the output of coalition $C.$\footnote{We normalize all bargaining problems so that the disagreement point is equal to the origin.}   An \textit{allocation} for the bargaining problem for $C$,  is a vector $x=(x_i)_{i\in C}\in \mathbb{R}_+^C$ such that $\sum_{i \in C}x_i=y(C).$ A \textit{bargaining rule} is a mapping that associates an allocation with each bargaining problem.

Given a coalitional bargaining problem $(U_{N},y(C)_{C\in \mathcal{K}})$, a bargaining rule $F$ induces a coalition formation game $(N, \succ_N)$ in the following way: for each $i \in N$ and each pair $C,C' \in \mathcal{K}$ with $i \in C\cap C',$ if $F_i(U_C, y(C))>F_i(U_{C'}, y(C'))$ then  $C \succ_i C'.$ Note that for the game to be well-defined, no pair of bargaining problems  should allocate the same amount to agent $i.$ 
A bargaining rule is \emph{pairwise aligned} if the coalition formation game induced is pairwise aligned for each bargaining problem.

\begin{corollary}\label{bargaining theorem}
Any coalition formation game induced by a pairwise aligned bargaining rule exhibits convergence to stability.
\end{corollary}
\begin{proof}
Let $(N,  \succ_N)$ be a coalition formation game induced by a pairwise aligned bargaining rule. \cite{pycia2012stability} guarantees that $(N,  \succ_N)$ is  a stable coalition formation game with no generalized rings.\footnote{\cite{pycia2012stability} shows that each pairwise aligned bargaining rule induces a stable coalition formation game (Corollary 1 in \cite{pycia2012stability}) with a rich domain of preferences. Lemmata 3 and 4 in \cite{pycia2012stability} state that a coalition formation game with rich domain and pairwise aligned preferences has no ``$n$-cycles in preferences''. The non-existence of ``$n$-cycles in preferences'' in his setting implies the non-existence of generalized rings in our setting.} By Proposition \ref{proposition5}, $(N,\succ_N)$ exhibits convergence to stability.
\end{proof}

Given $C \in \mathcal{K}$, the \emph{Nash bargaining rule} \citep{nash1950} for problem $(U_C, y(C))$ is determined by solving:
$$
\max_{x_{i}\geq 0}\prod_{i\in C}U_{i}(x)\text{ subject to }\sum_{i\in
C}x_{i}= y(C).
$$
The Nash bargaining rule  is included in the class of pairwise aligned bargaining rules \citep[see][p.331]{pycia2012stability} and therefore guarantees convergence to stability.
Another important bargaining solution is Kalai-Smorodinsky's \citep{kalai1975other}.
 Given 
$C \in \mathcal{K}$, the \emph{Kalai-Smorodinsky bargaining rule} for problem $(U_C, y(C))$ is determined by solving:
$$
\frac{U_{i}(x_{i})}{U_{i}(y(C))}=\frac{U_{j}(x_{j})}{U_{j}(y(C))}\text{ for all }i,j\in C \text{ subject to }
\sum\nolimits_{i\in C}x_{i}=y(C). 
$$
Not all games induced by the Kalai-Smorodinsky bargaining rule are stable. 
However, even if one considers only stable coalition formation games induced by the Kalai-Smorodinsky solution, it is found that they may lack convergence to stability. We show this in the following example.

\begin{example} Consider a risk-averse firm $f$ and a risk-neutral firm $g$ that can employ either one or two risk-averse workers $w_1,w_2$ whose utilities are given by
$$
U_{f}(x)=x^{1/4},~U_{g}=x,~U_{w_1}(x)=x^{1/6},~U_{w_2}(x)=x^{1/2}. 
$$
The following table provides the coalitions and the allocation given by the Nash and the Kalai-Smorodinsky (K-S) bargaining solutions for different levels of outputs: 

\medskip

{\small
\begin{center}
\ra{1.1}
\begin{tabular}{@{}ccccccc@{}}\toprule
Coalitions & $f~w_1w_2$ & $g \ w_1w_2$ & $ f \ w_1$ & $f \ w_2$ & $ g \ w_1$ & $g \ w_2$  \\ \midrule
Outputs & $43$ & $83$ & $20$ & $37$ & $1$ & $1$ \\ \cmidrule{1-7}
Nash & $(11.7, 7.8, 23.5)$ & $(49.8, 8.3, 24.9)$ & $(12,8)$ & $(12.3, 24.7)$ & $(0.8,0.2)$ & $(0.7,0.3)$\\ 
K-S & $(12.7, 6.9, 23.4)$ & $(49.6,3.8, 29.6)$ & $(11.4, 8.6)$ & $(14.1, 22.9)$ & $(0.8,0.2)$ & $(0.6,0.4)$\\ \bottomrule
\end{tabular}%
\end{center}}

\medskip

\noindent The  coalition formation game induced by Nash bargaining is:

{\small
\begin{center}

\ra{1.1}
\begin{tabular}{@{}ccccccc@{}}\toprule
$\boldsymbol{f}$ && $\boldsymbol{g}$ && $\boldsymbol{w_1}$ && $\boldsymbol{w_2}$  \\ 
\cmidrule{1-1} \cmidrule{3-3} \cmidrule{5-5} \cmidrule{7-7}
$f \ w_2$ && $\boldsymbol{g~w_1w_2}$ && $\boldsymbol{g~w_1w_2}$ && $\boldsymbol{g~w_1w_2}$  \\ 
$f \ w_1$ && $g \ w_1$  && $f \ w_1$  && $f \ w_2$  \\ 
$f \ w_1w_2$ && $g \ w_2$  && $f \  w_1w_2$  && $f \  w_1w_2$  \\ 
$\boldsymbol{f}$ &&  $g$ && $g \ w_1$  && $g \ w_2$  \\ 
&&   && $w_1$  && $w_2$  \\ \bottomrule
\end{tabular}%
\end{center}}

\noindent Observe that the unique reduced form is $\{\{f\},\{g~w_1w_2\}\}$ composed of a singleton and a fixed component implying not only stability but also convergence to stability.
The  coalition formation game induced by Kalai-Smorodinsky bargaining is:
{\small
\begin{center}
\ra{1.1}
\begin{tabular}{@{}ccccccc@{}}\toprule
$\boldsymbol{f}$ && $\boldsymbol{g}$ && $\boldsymbol{w_1}$ && $\boldsymbol{w_2}$  \\ 
\cmidrule{1-1} \cmidrule{3-3} \cmidrule{5-5} \cmidrule{7-7}
$f \ w_2$ && $g~w_1w_2$ && $f \ w_1$ && $g~w_1w_2$  \\ 
$\boldsymbol{f \ w_1w_2}$&& $g \ w_1$  &&$\boldsymbol{f \ w_1w_2}$  &&$\boldsymbol{f \ w_1w_2}$  \\ 
$f \ w_1$&& $g \ w_2$  && $g \  w_1w_2$  && $f \  w_2$  \\ 
$f$&& $\boldsymbol{g}$  && $g \ w_1$  && $g \ w_2$  \\
&&   && $w_1$  && $w_2$  \\ \bottomrule
\end{tabular}%
\end{center}}

\noindent Observe that there are two reduced forms: $\{\{f~w_1w_2\},\{ g\}\}$ and $\{\{f \ w_1, f \ w_2, g \ w_1w_2\}\}$. The first one is composed only of a fixed component and a singleton, guaranteeing stability. In contrast, the second one is composed of a generalized ring, implying lack of convergence to stability. 
\hfill $\Diamond$
\end{example}

\subsection{Coalition formation games and rationing rules}\label{rationing}

In the model considered by \cite{gallo2018rationing}, there is a set of agents with claims and
each coalition of agents produces an output which is insufficient to meet
the claims of its members.  Formally, given set of agents $N$ and a set of coalitions $\mathcal{K} \subseteq 2^N \setminus \{\emptyset\},$ a  \textit{coalitional rationing problem} 
is a tuple $(d_{N},y(C)_{C\in \mathcal{K}})$ where $d_{N}=(d_{i})_{i\in N}\in \mathbb{R}%
_{+}^{N}$ is a claims vector,  $y(C) \in \mathbb{R}_+$ is the
output of coalition $C$  and $%
\sum_{i\in C}d_{i}\geq y(C)$ for each $C\in \mathcal{K}.$ Given $C\in \mathcal{K}$, the \emph{rationing problem for $C$} is $(d_C, y(C))$ where $d_C=(d_i)_{i\in C}$ is the claims' vector of agents in $C$ and $y(C)$ is the output of coalition $C.$   An \textit{allocation} for the rationing problem $(d_C, y(C))$ is a vector $x=(x_i)_{i\in C}\in \mathbb{R}_+^C$ such that $\sum_{i \in C}x_i=y(C).$ A \textit{rationing rule} is a mapping that associates an allocation with each rationing problem. 

Given a coalitional rationing problem $(d_{N},y(C)_{C\in \mathcal{K}})$, a rationing rule $F$ induces a coalition formation game $(N, \mathcal{K}, \succ_N)$ in the following way: for each $i \in N$ and each pair $C,C' \in \mathcal{K}$ with $i \in C\cap C',$ if $F_i(d_C, y(C))>F_i(d_{C'}, y(C'))$ then  $C \succ_i C'.$ Note that for the game to be well-defined, no pair of rationing problems  should allocate the same amount to agent $i.$

One of the most important classes of rules for rationing problems is the class of parametric rules
\citep[see][]{young1987dividing,stovall2014collective}. The proportional, constrained equal
awards, constrained equal losses, and the Talmud and reverse Talmud rules
are symmetric parametric rules while the sequential priority rule is an
asymmetric parametric rule.

Let $f$ be a collection of functions $\{f_{i}\}_{i\in N}$,\footnote{%
When the rule is symmetric, $f_{i}$ is the same for all agents.} where each $%
f_{i} : \mathbb{R}_{+}\times [a,b]\longrightarrow  \mathbb{R}_{+}$ is continuous and weakly increasing in $\lambda ,$ $\lambda \in
[a,b]$, $-\infty \leq a<b\leq \infty $ and for each $i\in N$ and $%
d_{i}\in  \mathbb{R}_{+}$, $f_{i}(d_{i},a)=0$ and $\ f_{i}(d_{i},b)=d_{i}$. 
Given  $f$, a \emph{parametric (rationing) rule} $F$  is defined as follows. For each problem $(d,y)$ and each $i\in N$, 
\[
F_{i}(d,y)=f_{i}(d_{i},\lambda )\ \text{where }\ \lambda \ \text{ is\
chosen\ so\ that }\sum\nolimits_{i\in N}f_{i}(d_{i},\lambda )=y.\footnote{In the literature, $f$ is said to be a \textit{parametric representation} of $F$.} 
\]

\begin{corollary}\label{rationing theorem}
    Any coalition formation game induced by a parametric rule exhibits convergence to stability.
\end{corollary}
\begin{proof}
Let $(N, \succ_N)$ be a coalition formation game induced by a parametric rule. \cite{gallo2018rationing} guarantee that $(N,\succ_N)$ is  a stable coalition formation game with no generalized rings.\footnote{\cite{gallo2018rationing} show that each parametric rationing rule induces a stable coalition formation game with no ``rings in preferences'' \citep[Proposition 1 in][]{gallo2018rationing}. The non-existence of ``rings in preferences'' in their setting implies the non-existence of generalized rings in our setting.} By Proposition \ref{proposition5}, $(N, \succ_N)$ exhibits convergence to stability.
\end{proof}

 The random arrival rule
\citep{o1982problem} fails to guarantee stability. Moreover, focusing only on stable coalition formation games induced by the random arrival
rule, we find that they may lack convergence to stability. 
The following example illustrates the different behavior of the proportional rule and the random arrival rule  when inducing coalition formation games.\footnote{
For each $C\in \mathcal{K}$, each $(d_C,y(C))$, and each $i\in C$, 

\textbf{Proportional rule, $\boldsymbol{Prop}$:}   $$Prop_i (d_C,y(C))=\frac{d_i}{\sum_{j\in C}d_j}y(C). $$

\textbf{Random arrival rule, $\boldsymbol{RA}$:}   $$RA_i(d_C,y(C))=\frac{1}{|C|!}\left(  \sum_{\lessdot \in \mathcal{O}^C}\min\left\{ d_i, \max \left\{ y(C)-\sum_{j\in C,~j\lessdot i}d_j,0 \right\} \right\}\right),$$ where  $\mathcal{O}^C$ denote the class of strict orders on $C$, with generic element $\lessdot$.

}
\begin{example}\label{KS} Assume that there is a call to finance research projects and that several researchers are ready to submit a project. Each researcher has an aspiration, which depends on her CV, regarding the compensation she believes she deserves. Let  $N=\{1,2,3,4,5,6,7,8,9\}$  be the set of researchers with the following aspirations: 
\[
c_{1}=c_{2}=c_{5}=c_{7}=c_{8}=c_{9}=50,~c_{3}=c_{4}=c_{6}=10. 
\]%
Researchers can form various teams but participate in only one. Funding
depends on the quality of the project, which in turn depends on the team
composition, and there is not enough money to meet the aspirations of all possible teams. Assume that the money assigned to each potential team is distributed according to the random arrival and proportional rules. The table below shows the
coalitions, the outputs, and the distribution of the outputs obtained from these two rules.
\vspace{-10 pt}
{\small
\begin{center}
\ra{1.1}
\begin{tabular}{@{}ccccccccc@{}}
\toprule
Coalitions & $\{15\}$ & $\{45\}$ & $\{123\}$ & $\{34\}$ & $\{68\}$ & $%
\{78\}$ & $\{679\}$ & $\{26\}$ \\ 
\midrule

Outputs & $34$ & $20$ & $53$ & $9$ & $9$ & $34$ & $53$ & $20$ \\ 
\cmidrule{1-9}
$RA$ & $(17,17)$ & $(5,15)$ & $\left(\frac{73}{3},\frac{73}{3},\frac{13}{3}\right)$ & $\left(
\frac{9}{2},\frac{9}{2}\right)$ & $\left(\frac{9}{2},\frac{9}{2}\right)$ & $(17,17)$ & $\left(\frac{13}{3},\frac{73}{3},\frac{73}{3}\right)$ & $(15,5)$ \\ 
$Prop$ & $(17,17)$ & $(\frac{10}{3},\frac{50}{3})$ & $\left(\frac{265}{11},\frac{265}{11},\frac{53}{11}\right)$ & $\left(
\frac{9}{2},\frac{9}{2}\right)$ & $\left(\frac{3}{2},\frac{15}{2}\right)$ & $(17,17)$ &  $\left(\frac{53}{11},\frac{265}{11},\frac{265}{11}\right)$ & $(\frac{50}{3},\frac{10}{3})$ \\
\bottomrule
\end{tabular}
\end{center}}
\bigskip

\noindent The coalition formation game induced by random arrival rationing is:
{\small
\begin{center}
\ra{1.1}
\begin{tabular}{@{}ccccccccccccccccc@{}}\toprule
$\boldsymbol{1}$ &  & $\boldsymbol{2}$ && $\boldsymbol{3}$ && $\boldsymbol{4}$ && $\boldsymbol{5}$ && $\boldsymbol{6}$ && $\boldsymbol{7}$ && $\boldsymbol{8}$ && $\boldsymbol{9}$ \\
\cmidrule{1-1} \cmidrule{3-3} \cmidrule{5-5} \cmidrule{7-7} \cmidrule{9-9} \cmidrule{11-11} \cmidrule{13-13} \cmidrule{15-15} \cmidrule{17-17}
$123$ && $123$ && $\boldsymbol{34}$ && $45$ && $\boldsymbol{15}$ && $\boldsymbol{26}$ && $679$ && $\boldsymbol{78}$ && $679$ \\ 
$\boldsymbol{15}$ && $\boldsymbol{26}$ && $123$ && $\boldsymbol{34}$ && $45$ && $68$ && $\boldsymbol{78}$ && $68$ && $\boldsymbol{9}$ \\ 
 $1$ && $2$   &&$3$  && $4$  &&$5$   && $679$ && $7$  && $8$  &&   \\ 
  &&   &&  &&   &&   && $6$ && &&   &&   \\ 
 \bottomrule
\end{tabular}%
\end{center}}
\bigskip

\noindent Observe that there are two reduced forms: $\{\{15\},\{26\},\{34\},\{78\},\{9\}\}$ and $\{\{123\},\{45\},\{679,68,78\}\}$. The first one is composed of fixed components and a singleton, guaranteeing stability. In contrast, the second one is composed of two fixed components and a generalized ring, implying a lack of convergence to stability. 
The coalition formation game induced by proportional rationing is:

{\small
\begin{center}
\ra{1.1}
\begin{tabular}{@{}ccccccccccccccccc@{}}\toprule
$\boldsymbol{1}$ &  & $\boldsymbol{2}$ && $\boldsymbol{3}$ && $\boldsymbol{4}$ && $\boldsymbol{5}$ && $\boldsymbol{6}$ && $\boldsymbol{7}$ && $\boldsymbol{8}$ && $\boldsymbol{9}$ \\
\cmidrule{1-1} \cmidrule{3-3} \cmidrule{5-5} \cmidrule{7-7} \cmidrule{9-9} \cmidrule{11-11} \cmidrule{13-13} \cmidrule{15-15} \cmidrule{17-17}
$\boldsymbol{123}$ && $\boldsymbol{123}$ && $\boldsymbol{123}$ && $34$ && $15$ && $26$ && $\boldsymbol{679}$ && $78$ && $\boldsymbol{679}$ \\ 
$15$ && $26$ && $34$ && $\boldsymbol{45}$ && $\boldsymbol{45}$ && $\boldsymbol{679}$ && $78$ && $68$ && $9$ \\ 
 $1$ && $2$   &&$3$  && $4$  &&$5$   && $68$ && $7$  && $\boldsymbol{8}$  &&   \\ 
  &&   &&  &&   &&   && $6$ && &&   &&   \\ 
 \bottomrule
\end{tabular}%
\end{center}}

\noindent The only reduced form is $\{\{123\},\{45\},\{679\},\{8\}\}$, with no generalized ring.
\hfill $\Diamond$
\end{example}

\section{Final remarks}\label{Discusion}
Our paper contributes to the literature on coalition formation by characterizing absorbing sets in terms of reduced forms of the game.
We show that a reduced form condenses all the relevant information that can be extracted from an absorbing set without having to compute it. Furthermore, the reduced form of a game has the advantage over absorbing sets that the domination relation between coalition structures need not be considered.

Given a reduced form, we can reconstruct its associated absorbing set.
Each coalition structure of this set contains:
\begin{enumerate}[(i)]
    \item  a maximal set of each compact generalized ring,
    \item  a subset of a maximal set of each non-compact generalized ring,\footnote{There are coalition structures containing only one coalition of a non-compact generalized ring and others containing a maximal set of the coalitions in that generalized ring.} and
     \item all the fixed components.
\end{enumerate}
Agents gathered in the set of singletons of the reduced form appear in some coalition structures of the absorbing set forming non-singleton coalitions and in others as singletons.

In this paper, we also use the notion of reduced form to shed light on the problem of convergence to stability in coalition formation games for some economic environments. 

  We study stable coalition formation games satisfying properties such as common ranking, ordinary balance, weak top coalition, and pairwise alignment. Furthermore, stable coalition formation games induced by bargaining solutions and sharing rules are also considered. 

Finally, our approach opens up several interesting research directions, including the following: The paper relies on a dynamic process between coalition structures which is consistent with the standard blocking
definition in that all members of the blocking coalition become strictly better off,
and assumes that abandoned agents appear as singletons in the newly formed
coalition structure. However, another possibility is for the abandoned agents to
 get together as in, for instance, the $\delta$-model of \cite{hart1983endogenous} or the marriage model of \cite{tamura1993transformation}. How our notion of reduced form adapts to this new dynamic is an open question.

Another interesting direction to explore is considering indifferences in preference relations. We think such an adaptation is not straightforward but could be accomplished by carefully adjusting several of the definitions involved in our analysis. We leave this interesting extension for further research.


\begin{thebibliography}{38}
\newcommand{\enquote}[1]{``#1''}
\expandafter\ifx\csname natexlab\endcsname\relax\def\natexlab#1{#1}\fi

\bibitem[\protect\citeauthoryear{Banerjee, Konishi, and S{\"o}nmez}{Banerjee
  et~al.}{2001}]{banerjee2001core}
\textsc{Banerjee, S., H.~Konishi, and T.~S{\"o}nmez} (2001): \enquote{Core in a
  simple coalition formation game,} \emph{Social Choice and Welfare}, 18,
  135--153.

\bibitem[\protect\citeauthoryear{Barber{\`a}, Bevi{\'a}, and
  Ponsat{\'\i}}{Barber{\`a} et~al.}{2015}]{barbera2015meritocracy}
\textsc{Barber{\`a}, S., C.~Bevi{\'a}, and C.~Ponsat{\'\i}} (2015):
  \enquote{Meritocracy, egalitarianism and the stability of majoritarian
  organizations,} \emph{Games and Economic Behavior}, 91, 237--257.

\bibitem[\protect\citeauthoryear{Bogomolnaia and Jackson}{Bogomolnaia and
  Jackson}{2002}]{bogomolnaia2002stability}
\textsc{Bogomolnaia, A. and M.~O. Jackson} (2002): \enquote{The stability of
  hedonic coalition structures,} \emph{Games and Economic Behavior}, 38,
  201--230.

\bibitem[\protect\citeauthoryear{Bondareva}{Bondareva}{1963}]{bondareva1963some}
\textsc{Bondareva, O.~N.} (1963): \enquote{Some applications of linear
  programming methods to the theory of cooperative games,} \emph{Problemy
  Kibernetiki}, 10, 119--139.

\bibitem[\protect\citeauthoryear{Chung}{Chung}{2000}]{chung2000existence}
\textsc{Chung, K.-S.} (2000): \enquote{On the existence of stable roommate
  matchings,} \emph{Games and Economic Behavior}, 33, 206--230.

\bibitem[\protect\citeauthoryear{Demuynck, Herings, Saulle, and Seel}{Demuynck
  et~al.}{2019}]{demuynck2019myopic}
\textsc{Demuynck, T., P.~J.-J. Herings, R.~D. Saulle, and C.~Seel} (2019):
  \enquote{The myopic stable set for social environments,} \emph{Econometrica},
  87, 111--138.

\bibitem[\protect\citeauthoryear{Diamantoudi, Miyagawa, and Xue}{Diamantoudi
  et~al.}{2004}]{diamantoudi2004random}
\textsc{Diamantoudi, E., E.~Miyagawa, and L.~Xue} (2004): \enquote{Random paths
  to stability in the roommate problem,} \emph{Games and Economic Behavior},
  48, 18--28.

\bibitem[\protect\citeauthoryear{Echenique and Yenmez}{Echenique and
  Yenmez}{2007}]{echenique2007solution}
\textsc{Echenique, F. and M.~B. Yenmez} (2007): \enquote{A solution to matching
  with preferences over colleagues,} \emph{Games and Economic Behavior}, 59,
  46--71.

\bibitem[\protect\citeauthoryear{Eriksson and H{\"a}ggstr{\"o}m}{Eriksson and
  H{\"a}ggstr{\"o}m}{2008}]{eriksson2008instability}
\textsc{Eriksson, K. and O.~H{\"a}ggstr{\"o}m} (2008): \enquote{Instability of
  matchings in decentralized markets with various preference structures,}
  \emph{International Journal of Game Theory}, 36, 409--420.

\bibitem[\protect\citeauthoryear{Farrell and Scotchmer}{Farrell and
  Scotchmer}{1988}]{farrell1988partnerships}
\textsc{Farrell, J. and S.~Scotchmer} (1988): \enquote{Partnerships,} \emph{The
  Quarterly Journal of Economics}, 103, 279--297.

\bibitem[\protect\citeauthoryear{Gale and Shapley}{Gale and
  Shapley}{1962}]{gale1962college}
\textsc{Gale, D. and L.~S. Shapley} (1962): \enquote{College admissions and the
  stability of marriage,} \emph{The American Mathematical Monthly}, 69, 9--15.

\bibitem[\protect\citeauthoryear{Gallo and Inarra}{Gallo and
  Inarra}{2018}]{gallo2018rationing}
\textsc{Gallo, O. and E.~Inarra} (2018): \enquote{Rationing rules and stable
  coalition structures,} \emph{Theoretical Economics}, 13, 933--950.

\bibitem[\protect\citeauthoryear{Greenberg and Weber}{Greenberg and
  Weber}{1986}]{greenberg1986strong}
\textsc{Greenberg, J. and S.~Weber} (1986): \enquote{Strong Tiebout equilibrium
  under restricted preferences domain,} \emph{Journal of Economic Theory}, 38,
  101--117.

\bibitem[\protect\citeauthoryear{Hart and Kurz}{Hart and
  Kurz}{1983}]{hart1983endogenous}
\textsc{Hart, S. and M.~Kurz} (1983): \enquote{Endogenous formation of
  coalitions,} \emph{Econometrica}, 1047--1064.

\bibitem[\protect\citeauthoryear{Herings, Saulle, and Seel}{Herings
  et~al.}{2021}]{herings2021last}
\textsc{Herings, P. J.-J., R.~D. Saulle, and C.~Seel} (2021): \enquote{The last
  will be first, and the first last: Segregation in societies with relative
  pay-off concerns,} \emph{The Economic Journal}, 131, 2119--2143.

\bibitem[\protect\citeauthoryear{Iehl{\'e}}{Iehl{\'e}}{2007}]{iehle2007core}
\textsc{Iehl{\'e}, V.} (2007): \enquote{The core-partition of a hedonic game,}
  \emph{Mathematical Social Sciences}, 54, 176--185.

\bibitem[\protect\citeauthoryear{Inal}{Inal}{2015}]{inal2015core}
\textsc{Inal, H.} (2015): \enquote{Core of coalition formation games and
  fixed-point methods,} \emph{Social Choice and Welfare}, 45, 745--763.

\bibitem[\protect\citeauthoryear{Inarra, Larrea, and Molis}{Inarra
  et~al.}{2013}]{inarra2013absorbing}
\textsc{Inarra, E., C.~Larrea, and E.~Molis} (2013): \enquote{Absorbing sets in
  roommate problems,} \emph{Games and Economic Behavior}, 81, 165--178.

\bibitem[\protect\citeauthoryear{Jackson and Watts}{Jackson and
  Watts}{2002}]{jackson2002evolution}
\textsc{Jackson, M.~O. and A.~Watts} (2002): \enquote{The evolution of social
  and economic networks,} \emph{Journal of Economic Theory}, 106, 265--295.

\bibitem[\protect\citeauthoryear{Kalai and Schmeidler}{Kalai and
  Schmeidler}{1977}]{kalai1977admissible}
\textsc{Kalai, E. and D.~Schmeidler} (1977): \enquote{An admissible set
  occurring in various bargaining situations,} \emph{Journal of Economic
  Theory}, 14, 402--411.

\bibitem[\protect\citeauthoryear{Kalai and Smorodinsky}{Kalai and
  Smorodinsky}{1975}]{kalai1975other}
\textsc{Kalai, E. and M.~Smorodinsky} (1975): \enquote{Other solutions to
  Nash’s bargaining problem,} \emph{Econometrica}, 43, 513--518.

\bibitem[\protect\citeauthoryear{Klaus and Klijn}{Klaus and
  Klijn}{2005}]{klaus2005stable}
\textsc{Klaus, B. and F.~Klijn} (2005): \enquote{Stable matchings and
  preferences of couples,} \emph{Journal of Economic Theory}, 121, 75--106.

\bibitem[\protect\citeauthoryear{K{\'o}czy and Lauwers}{K{\'o}czy and
  Lauwers}{2004}]{koczy2004coalition}
\textsc{K{\'o}czy, L.~{\'A}. and L.~Lauwers} (2004): \enquote{The coalition
  structure core is accessible,} \emph{Games and Economic Behavior}, 48,
  86--93.

\bibitem[\protect\citeauthoryear{Kojima and {\"U}nver}{Kojima and
  {\"U}nver}{2008}]{kojima2008random}
\textsc{Kojima, F. and M.~U. {\"U}nver} (2008): \enquote{Random paths to
  pairwise stability in many-to-many matching problems: a study on market
  equilibration,} \emph{International Journal of Game Theory}, 36, 473--488.

\bibitem[\protect\citeauthoryear{Nash}{Nash}{1950}]{nash1950}
\textsc{Nash, J.} (1950): \enquote{The bargaining problem,}
  \emph{Econometrica}, 28, 155--162.

\bibitem[\protect\citeauthoryear{Olaizola and Valenciano}{Olaizola and
  Valenciano}{2014}]{olaizola2014asymmetric}
\textsc{Olaizola, N. and F.~Valenciano} (2014): \enquote{Asymmetric flow
  networks,} \emph{European Journal of Operational Research}, 237, 566--579.

\bibitem[\protect\citeauthoryear{O'Neill}{O'Neill}{1982}]{o1982problem}
\textsc{O'Neill, B.} (1982): \enquote{A problem of rights arbitration from the
  Talmud,} \emph{Mathematical Social Sciences}, 2, 345--371.

\bibitem[\protect\citeauthoryear{Pycia}{Pycia}{2012}]{pycia2012stability}
\textsc{Pycia, M.} (2012): \enquote{Stability and preference alignment in
  matching and coalition formation,} \emph{Econometrica}, 80, 323--362.

\bibitem[\protect\citeauthoryear{Ray}{Ray}{2007}]{ray2007game}
\textsc{Ray, D.} (2007): \emph{A game-theoretic perspective on coalition
  formation}, Oxford University Press.

\bibitem[\protect\citeauthoryear{Ray and Vohra}{Ray and
  Vohra}{2015}]{ray2015coalition}
\textsc{Ray, D. and R.~Vohra} (2015): \enquote{Coalition formation,}
  \emph{Handbook of game theory with economic applications}, 4, 239--326.

\bibitem[\protect\citeauthoryear{Roth}{Roth}{2018}]{roth2018marketplaces}
\textsc{Roth, A.~E.} (2018): \enquote{Marketplaces, markets, and market
  design,} \emph{American Economic Review}, 108, 1609--58.

\bibitem[\protect\citeauthoryear{Roth and Vande~Vate}{Roth and
  Vande~Vate}{1990}]{roth1990random}
\textsc{Roth, A.~E. and J.~H. Vande~Vate} (1990): \enquote{Random paths to
  stability in two-sided matching,} \emph{Econometrica}, 1475--1480.

\bibitem[\protect\citeauthoryear{Schwartz}{Schwartz}{1970}]{schwartz1970possibility}
\textsc{Schwartz, T.} (1970): \enquote{On the possibility of rational policy
  evaluation,} \emph{Theory and Decision}, 1, 89--106.

\bibitem[\protect\citeauthoryear{Shapley}{Shapley}{1967}]{shapley1967balanced}
\textsc{Shapley, L.~S.} (1967): \enquote{On balanced sets and cores,}
  \emph{Naval Research Logistics Quarterly}, 14, 453--460.

\bibitem[\protect\citeauthoryear{Stovall}{Stovall}{2014}]{stovall2014collective}
\textsc{Stovall, J.~E.} (2014): \enquote{Collective rationality and monotone
  path division rules,} \emph{Journal of Economic Theory}, 154, 1--24.

\bibitem[\protect\citeauthoryear{Tamura}{Tamura}{1993}]{tamura1993transformation}
\textsc{Tamura, A.} (1993): \enquote{Transformation from arbitrary matchings to
  stable matchings,} \emph{Journal of Combinatorial Theory, Series A}, 62,
  310--323.

\bibitem[\protect\citeauthoryear{Tan}{Tan}{1991}]{tan1991necessary}
\textsc{Tan, J.~J.} (1991): \enquote{A necessary and sufficient condition for
  the existence of a complete stable matching,} \emph{Journal of Algorithms},
  12, 154--178.

\bibitem[\protect\citeauthoryear{Young}{Young}{1987}]{young1987dividing}
\textsc{Young, H.~P.} (1987): \enquote{On dividing an amount according to
  individual claims or liabilities,} \emph{Mathematics of Operations Research},
  12, 398--414.

\end{thebibliography}

\appendix
\section{Relation between rings and cycles}\label{section ring en preferences}

This appendix studies the relationship between two well-known concepts in the literature: Cycles of coalition structures and rings of coalitions.
We show that each cycle induces a ring and each ring induces a cycle. This result links non-trivial absorbing sets with generalized rings in the following way: Each non-trivial absorbing set consists of one or more cycles, each cycle induces a ring, and a generalized ring can be constructed by merging all overlapping rings.

A cycle of coalition structures is an ordered set of coalition structures that shows cyclical behavior. That is, for each pair of consecutive coalition structures of the ordered set, the successor coalition structure dominates its predecessor.\footnote{Here, given an ordered set of coalition structures with $k$ elements, two coalition structures are \emph{consecutive} if their indices differ by one (modulo $k$).} Formally,  
\begin{definition}
An ordered set of coalition structures $(\pi_1,\ldots,\pi_{J})\subset \Pi,$ with $J \geq 3,$
 is a \textbf{cycle}  if $\pi_{j+1} \gg \pi_{j}$ for $j=1,\ldots,J$ subscript modulo $J$.
\end{definition}

A ring of coalitions is an ordered set of non-singleton coalitions that behaves cyclically, i.e.,  for each pair of consecutive coalitions\footnote{Here, given an ordered set of coalitions with $k$ elements, two coalitions are \emph{consecutive} if their indices differ by one (modulo $k$). Note that our concept of consecutive coalition is not related with the notion presented in \cite{greenberg1986strong}.} of the ordered set the successor coalition is preferred to its predecessor \citep[see][among others]{diamantoudi2004random,inarra2013absorbing,tan1991necessary}.\footnote{There are other ways of defining cyclicity between coalitions  \citep[see][ for more details]{pycia2012stability,inal2015core}.}  Formally, 
\begin{definition}\label{definicion de ring}
 An ordered set of non-singleton coalitions $(R_1, \ldots,R_J)\subseteq \mathcal{K}$, with $J\geq 3$,  is a \textbf{ring} if 
$R_{j+1}\succ R_{j}$ for $j=1,\ldots,J$ subscript modulo $J$.
\end{definition}


Next,  we present an algorithm that constructs a ring of coalitions from a cycle of coalition structures. 
Let $\! \mathcalligra{C}\; =(\pi_1,\ldots,\pi_{J}) $ be a cycle of coalition structures, let  $C_j$ denote the coalition that is formed in $\pi_j,$ i.e.  $\pi_j \gg \pi_{j-1}$ via $C_j$, and consider the ordered set $\mathcal{C}=(C_1, \ldots,C_{J}).$ 
To construct a ring, proceed as follows:\bigskip

\begin{center}
\begin{tabular}{l l}
\hline \hline
\multicolumn{2}{l}{\textbf{Algorithm:}}\vspace*{10 pt}\\

\textbf{Step 1} & Set $\overline{R}_1$ as any coalition in $\mathcal{C}$.  \\
\textbf{Step $\boldsymbol{t}$} &  Set \\
& $\overline{R}_t \equiv \min_{r\geq 1}\{C_{j+r} \text{ such that } C_j=\overline{R}_{t-1} \text{ and } C_j \cap C_{j+r} \neq \emptyset$ \\
&\hspace{70 pt} $\text{ with } j+r  \text{ mod }J\}.$\\
& \texttt{IF} $\overline{R}_t=\overline{R}_s$ for $s<t,$ \\
& \hspace{20 pt}\texttt{THEN} set $(\overline{R}_{s+1}, \ldots, \overline{R}_t),$ and \texttt{STOP}. \\
& \texttt{ELSE} continue to Step $t+1.$ \\
\hline \hline
\end{tabular}
\end{center}
\bigskip

\noindent Notice that in each step of the algorithm, a different coalition of $\mathcal{C}$ is selected except in the last step, where one of the previously selected coalitions is singled out. Therefore, the algorithm stops in at most $J+1$ steps (recall that $J=|\mathcal{C}|$).  

The following proposition establishes that the ordered set $(\overline{R}_{s+1}, \ldots, \overline{R}_t),$ where the above algorithm determines $s$, is a ring. Conversely, for each ring, a cycle of coalition structures can be constructed. This result is crucial for the proof of the main finding of the paper, Theorem \ref{bijection}, presented in Appendix \ref{apendice}.
\begin{proposition}\label{Th ring cycle}
A coalition formation game has a ring of coalitions if and only if it has a cycle of coalition structures.
\end{proposition}
\begin{proof}
$(\Longleftarrow)$
Let $\! \mathcalligra{C} \; \ $ be a cycle of coalition structures. The application of the above algorithm results in the ordered set $(\overline{R}_{s+1}, \ldots, \overline{R}_t)$. To simplify notation, we rename the elements of the ordered set and write $(R_1, \ldots, R_\ell)=(\overline{R}_{s+1}, \ldots, \overline{R}_t).$ We claim that the ordered set $(R_1, \ldots, R_\ell)$ thus constructed is a ring, i.e. for each $R_{j+1}$ and $R_j$ in the ordered set, $R_{j+1}\succ R_j$ and $\ell\geq 3$. Take any coalition $R_j.$ Coalition $R_{j+1}$ (modulo $\ell$) is the closest coalition that has a non-empty intersection with $R_j$ (following the modular order of the coalition structures in cycle $\! \mathcalligra{C} \ \ \,$), so all the coalition structures between the one in which $R_j$ breaks and the one in which $R_{j+1}$ breaks contain coalition $R_j$.
Let $\pi$ and $\pi'$ be the two consecutive coalition structures in $\! \mathcalligra{C}  \ $  such that $\pi' \gg \pi$ via $R_{j+1}$. $R_{j+1}$ is the breaking coalition, so $R_{j+1}$ belongs to $\pi'$. Furthermore,  since $R_j$ belongs to $\pi$ and  $R_{j+1}\cap R_j\neq \emptyset$, by definition of the domination relation $\gg$,  $R_{j+1}\succ R_j$. Furthermore, 
$\ell \geq 3.$ This holds for the following two facts: (i) there are at least two coalitions in the ordered set because all the coalitions that break in a cycle are also broken; (ii) if there are only two coalitions, say $R_1$ and $R_2,$ then there is an agent $i \in R_1\cap R_2$ such that $R_1 \succ_i R_2 \succ_i R_1,$ which by transitivity implies $R_1 \succ_i R_1,$ which is a contradiction. Therefore, $(R_1, \ldots, R_\ell)$ is a ring.

\noindent $(\Longrightarrow)$ Let $(R_1,\ldots,R_J)$ be a ring in  coalition formation game $(N,\succ_N)$. Define the collection of coalition structures $\! \mathcalligra{C} \; =(\pi_1,\ldots,\pi_J)$ where $\pi_j$ is as follows:
$$\pi_j (i)=\left\{ 
\begin{tabular}{ll}
$R_j$ & $\text{for }i \in R_j $ \\
$\{i\}$ & otherwise.%
\end{tabular}%
\right.
$$
By Definition \ref{definicion de ring},  $\pi_j \gg\pi_{j-1}$ via $R_j$ for each $j=1,\ldots,J$ (subscript modulo $J$). Therefore, $\! \mathcalligra{C}\ \  $ is a cycle. \end{proof}

Next, we illustrate the above result with an example.
\medskip

\noindent \textbf{Example 1 (Continued)} \textit{Consider the ring $(15,12,23,34,45)$ in Example \ref{ejemplo ring}.  The collection $\mathcalligra{C}=\left(\{12,34,5\}, \{12,3,45\},\{1,23,45\},\{15,23,4\},\{15,2,34\}\right)$ is a cycle of coalition structures. Starting from $\{12,34,5\}$, the set of
blocking coalitions between coalition structures is $\mathcal{C}%
=(45,23,15,34,12)$. Assume that Step 1 of the above algorithm selects
coalition $45$. The following steps select coalitions $15$, $12,$ $23$ and $34$,
respectively. The algorithm ends when coalition $45$ is reached again and
ring $(45,15,12,23,34)$ is obtained.  \hfill $\Diamond$}
\medskip

Given a non-trivial absorbing set $\mathcal{A}$, by Definition \ref{absorbing}, there is a collection of cycles forming  $\mathcal{A}$. For each cycle, by Proposition \ref{Th ring cycle}, there is a collection of rings. We say that each one of these rings is \textit{derived from} $\mathcal{A}.$ Thus, by merging all the overlapping rings derived from $\mathcal{A}$, we can construct all generalized rings that belong to the reduced form associated with $\mathcal{A}.$ In the following lemmata we formalize this construction.
\begin{lemma}\label{remark relacion entre rings y generalized rings}
Let $\mathcal{A}$ be a non-trivial absorbing set and consider a maximal collection of overlapping rings derived from $\mathcal{A}$. If  $\mathcal{B}$ is the set of all coalitions that belong to such rings, then $\mathcal{B}$ is a generalized ring. 
\end{lemma}
\begin{proof}
Let $\mathcal{A}$ be a non-trivial absorbing set and consider a maximal collection of rings derived from $\mathcal{A}$ that overlaps. Let  $\mathcal{B}$ be the set of all coalitions that belong to those rings.  We now show that $\mathcal{B}$ fulfills Conditions (i) and (ii) of Definition \ref{def generalize ring}. By construction of the overlapping rings, Condition (i) is fulfilled straightforwardly. To see Condition (ii), assume that there are a coalition $\widetilde{C} \in \mathcal{B}$ and a coalition structure $\widetilde{\pi}\in \mathcal{A}$ such that $\widetilde{C}\in \widetilde{\pi}$.  There are two cases to consider:
\begin{enumerate}
\item [$\boldsymbol{1.}$]\textbf{There is $\boldsymbol{\mathcal{M} \in \mathcalligra{M}_{\mathcal{B}}$ such that $\mathcal{M}\subseteq \widetilde{\pi}}$}. Thus, by construction of $\mathcal{B}$, there is $C'\in \mathcal{B}$ with $C' \succ \widetilde{C}$ such that $C'$ breaks $\widetilde{\pi}$ and, therefore, $C'$ breaks $\mathcal{M}$. 
\item [$\boldsymbol{2.}$]\textbf{For each  $\boldsymbol{\mathcal{M} \in \mathcalligra{M}_{\mathcal{B}}, \mathcal{M}\nsubseteq \widetilde{\pi}}.$} Take the set of coalitions  $\mathcal{B}\cap \widetilde{\pi}$.  By definition of maximal set, there is $\widetilde{\mathcal{M}} \in \mathcalligra{M}_{\mathcal{B}}$ such  $\mathcal{B}\cap \widetilde{\pi} \subseteq \widetilde{\mathcal{M}}.$ Note that the agents in $N(\mathcal{B})\setminus N(\mathcal{B}\cap \widetilde{\pi})$ are singletons in $\widetilde{\pi}.$ Thus, there is a $\widehat{\pi}$ such that $\widetilde{\mathcal{M}} \subseteq \widehat{\pi}$ and $\widehat{\pi}\gg^T \widetilde{\pi}$ by forming each coalition in $\widetilde{\mathcal{M}}\setminus \widetilde{\pi}$ with agents of $N(\mathcal{B})$ that are singletons in $\widetilde{\pi}$. Now, by Case 1, there is a coalition $C'$ that breaks $\widetilde{\mathcal{M}}.$
\end{enumerate}
Since Conditions (i) and (ii) of Definition \ref{def generalize ring} are fulfilled, $\mathcal{B}$ is a generalized ring.\end{proof}

In this way, for a given non-trivial absorbing set, we can construct a collection of generalized rings. Formally, 

\begin{lemma}\label{construyo todos los RC}
A non-trivial absorbing set induces a collection of generalized rings.
\end{lemma}
\begin{proof}
Let $\mathcal{A}$ be a non-trivial absorbing set. Notice that, given any two different coalition structures in $\mathcal{A}$, by Definition \ref{absorbing} there is a cycle of coalition structures in $\mathcal{A}$ that includes those coalition structures. $\mathcal{A}$ can, therefore, be seen as the union of all such cycles. Thus, by Proposition \ref{Th ring cycle}, for each cycle of coalition structures in $\mathcal{A}$ there is a ring. Thus, by Lemma \ref{remark relacion entre rings y generalized rings}, all generalized rings induced by $\mathcal{A}$ are constructed. 
\end{proof}

\section{Proof of Theorem \ref{bijection}}\label{apendice}

This appendix proves that each absorbing set can be identified with a reduced form. To prove that each reduced form $\mathcalligra{R}\ $ generates an absorbing set, we first define a special class of coalition structures constructed from $\mathcalligra{R}\ $, that we call $\mathcalligra{R}\ $ - coalition structures (Definition \ref{definicion de D coalition structure}). We show that each $\mathcalligra{R}\ $ - coalition structure transitively dominates: (i) all other $\mathcalligra{R}\ $ - coalition structures (Lemma \ref{pisubde}), and (ii) all coalition structures that contain all non-singleton coalitions of that $\mathcalligra{R}\ $ - coalition structure (Lemma \ref{compacta dominada}). We also show that if a coalition structure transitively dominates a $\mathcalligra{R}\ $ - coalition structure, then the converse also follows (Lemma \ref{pisubde bis}).

Each $\mathcalligra{R}\ $ - coalition structure, in turn, defines a set: The set that includes it together with all coalition structures that transitively dominate it. By the previous results, we show that this set only depends on  $\mathcalligra{R}\ $ (and not on the particular $\mathcalligra{R}\ $ - coalition structure chosen to construct it).  This set is proven to be an absorbing set (Proposition \ref{absorbing and stable decomp}).

To show that each absorbing set generates a reduced form, first we observe that an absorbing set is formed by overlapping cycles of coalition structures. From the collection of all these cycles, we can construct all generalized rings derived from such absorbing set (Lemma \ref{construyo todos los RC}). Furthermore, we identify the fixed components by considering the coalitions that appear in all coalition structures of that absorbing set. Thus, we construct a collection of generalized rings, fixed components, and a set of singletons consisting of the remaining agents. This collection turns out to be a reduced form (Theorem \ref{bijection}).

Now we start by formalizing the definition of  a $\mathcalligra{R}\ $ - coalition structure.\medskip

\begin{definition}\label{definicion de D coalition structure} Let $(N,\succ_N)$ be a coalition formation game and  $\mathcalligra{R} \ $ \   a reduced form of that game. An $\mathcalligra{R}$ \ \textbf{-- coalition structure}  is  a coalition structure $\pi_{\mathcalligra{R}}$  \  such that:

\begin{enumerate}[(i)]

\item for each compact generalized ring or each fixed component $\mathcal{B} \in \mathcalligra{R} \ $, $\pi_{\mathcalligra{R}}\  \cap\mathcal{B}=\mathcal{M}$ for some $\mathcal{M} \in \mathcalligra{M} \, _\mathcal{B}$, and  $\pi_{\mathcalligra{R}}\ (i)=\{i\}$ for each $i\in N(\mathcal{B}) \setminus N(\mathcal{M}).$\footnote{Given a maximal set  $\mathcal{M}$, denote by $N(\mathcal{M})$  the set of agents that belong to (at least) one coalition in $\mathcal{M},$ that is, $N(\mathcal{M}) \equiv \bigcup_{C \in \mathcal{M}}C.$}

\item for each non-compact generalized ring $\mathcal{B} \in \mathcalligra{R} \ $, $\pi_{\mathcalligra{R}}\  \cap\mathcal{B}=C$ for some $C \in \mathcal{B}$, and  $\pi_{\mathcalligra{R}}\ (i)=\{i\}$ for each $i\in N(\mathcal{B}) \setminus C.$


\item for each $i \in \mathcal{S} $, $\pi_{\mathcalligra{R}}\ (i)=\{i\}.$

\end{enumerate}
\end{definition}\medskip

\noindent Condition (i) says that for  each compact generalized ring or each fixed component of $\mathcalligra{R} \ $, an $\mathcalligra{R}\ $ -- coalition structure includes: 
\begin{itemize}
    \item a maximal set of that generalized ring,
    \item the  agents of that generalized ring not involved in that maximal set as singletons, and
    \item the coalition of each fixed component.
 \end{itemize}
Condition (ii) says that for each non-compact generalized ring of $\mathcalligra{R}\ $, an $\mathcalligra{R}\ $ -- coalition structure includes: 
\begin{itemize}
    \item only one coalition of that generalized ring, and
    \item the agents of that generalized ring not involved in the coalition as singletons.
\end{itemize}
 Condition (iii) says that each agent in $\mathcal{S}$ is a singleton in an $\mathcalligra{R}$ \ -- coalition structure.

\begin{lemma}\label{pisubde}
Let $\mathcalligra{R}$ \  \ be a reduced form and let  $\pi_{ \!\! \mathcalligra{R}} \ $ and $\pi'_{\!\!\mathcalligra{R}} \ $ be two different $\mathcalligra{R}$ -- coalition structures. Then  $ \pi_{\!\!\mathcalligra{R}} \ \gg ^T \pi'_{\!\!\mathcalligra{R}} \ .$ 
\end{lemma}
\begin{proof}
Let $\mathcalligra{R}$ \  \ be a reduced form and let $\pi_{ \!\! \mathcalligra{R}} \ $ and $\pi'_{\!\!\mathcalligra{R}} \ $ be two different $\mathcalligra{R}$ -- coalition structures.  By Definition \ref{def generalize ring}, there are a sequence of coalitions $C_1,\ldots,C_J$  and a sequence of coalition structures $\pi_0,\ldots,\pi_J$  such that:
\begin{enumerate}[(i)]
\item $\pi_0=\pi'_{\!\!\mathcalligra{R}} \ $ and $\pi_J = \pi_{\mathcalligra{R}}$\ ,
\item $C_j\in \mathcal{B}$ for some $\mathcal{B} \in \mathcalligra{R}\ \ $ and $C_{j}\cap C\neq \emptyset$ for some $C \in \mathcal{K}\cap \pi_{j-1}$ for $j=1,\ldots,J$, \item $\pi_j \gg \pi_{j-1}$ via $C_j$ for each $j=1,\ldots,J.$
\end{enumerate}    
Therefore,  $\pi_{\!\!\mathcalligra{R}} \ \gg^T \pi'_{\mathcalligra{R}}$ \ . 
\end{proof}

The domination relation between coalition structures when a non-compact generalized ring is involved has a particular feature: When a maximal set of a non-compact generalized ring is included in a coalition structure, there is another coalition structure that contains only one coalition of that generalized ring that transitively dominates only one coalition in the maximal set. For instance, consider Example \ref{ejemplo tipo 2}, which only has one non-compact generalized ring. In this example $\{145,23,6,78\} \gg \{1,23,46,5,78\}$ via $145$, $\{1,2,356,4,78\} \gg \{145,23,6,78\}$ via $356$, and $\{1,2,3,46,$ $5,78\} \gg \{1,2,356,4,78\}$ via $46$. Therefore,  $\{1,2,3,46,5,78\} \gg^T \{1,23,46,5,78\}.$ This ``disintegrating'' behavior of maximal sets is always present when the generalized rings are non-compact. The following lemma deals with this fact and is used to prove Lemma \ref{compacta dominada}.
\begin{lemma}\label{reducir a compact collection}
Let $\mathcal{B}$ be a non-compact generalized ring,  $\mathcal{M} \in \mathcalligra{M}_{\mathcal{B}}$, and  let $\pi \in \Pi$  be such that $\mathcal{M}\subseteq \pi
$. 
 If  $C$ is any coalition in $\mathcal{B}$,  and $\pi^{\star} \in \Pi$ is such that
$$\pi^{\star} (i)=\left\{ 
\begin{tabular}{ll}
$C$ & $\text{for } i\in C$ \\
$\pi(i)$ & $\text{for }i \in N \setminus  N(\mathcal{B})$\\
$\{i\}$ & otherwise, \\%
\end{tabular}%
\right.
$$
 then $\pi^{\star} \gg^{T} \pi.$

\end{lemma}
\begin{proof}
Since $\mathcal{B}$ is non-compact, there are $\widetilde{C} \in \mathcal{B}$ and  $\widetilde{\mathcal{M}}\in \mathcalligra{M}_{\mathcal{B}}$ such that $\widetilde{C}$ has non-empty intersection with  at least two coalitions of $\widetilde{\mathcal{M}}$.
Let $\widetilde{\pi}\in \Pi$ be such that 
$$\widetilde{\pi}(i)=\left\{ 
\begin{tabular}{ll}
$\widetilde{\mathcal{M}}(i)$ & $\text{for } i\in N(\widetilde{\mathcal{M}})$ \\
$\pi(i)$ & $\text{for }i \in N \setminus  N(\mathcal{B})$\\
$\{i\}$ & otherwise. \\%
\end{tabular}%
\right.
$$
Thus, by Definition \ref{def generalize ring}, we have that $\widetilde{\pi} \gg^{T} \pi.$

Assume that  $\widetilde{C}$  has non-empty intersection with  each coalition in $\widetilde{\mathcal{M}}$. Let  $\pi'\in \Pi$ be such that
$$\pi'(i)=\left\{ 
\begin{tabular}{ll}
$\widetilde{C}$ & $\text{for } i\in \widetilde{C}$\\
$\pi(i)$ & $\text{for }i \in N \setminus  N(\mathcal{B})$\\
$\{i\}$ & otherwise. \\%
\end{tabular}%
\right.
$$
Thus,  $\pi' \gg \widetilde{\pi}$ via $\widetilde{C}$. By Definition \ref{def generalize ring}, $\pi^{\star} \gg^T \pi'$. Therefore, $\pi^{\star} \gg^T \pi'\gg \widetilde{\pi}\gg^T \pi$ implying $\pi^{\star} \gg^T  \pi$, and the proof is complete. 

Assume now that $\widetilde{C}$  has a non-empty intersection with  $k$ coalitions of $\widetilde{\mathcal{M}}$ with $k<|\widetilde{\mathcal{M}}|$. Let $\mathcal{E}\subseteq \widetilde{\mathcal{M}}$ be such that each coalition in $\mathcal{E}$ is disjoint with $\widetilde{C}$. Thus, $|\mathcal{E}|\geq 1.$ Let $\pi_{\mathcal{E}}\in \Pi$ be such that 
$$\pi_{\mathcal{E}}(i)=\left\{ 
\begin{tabular}{ll}
$\widetilde{C}$ & $\text{for } i\in \widetilde{C}$\\
$\mathcal{E}(i)$ & $\text{for } i\in N(\mathcal{E})$ \\
$\pi(i)$ & $\text{for }i \in N \setminus  N(\mathcal{B})$\\
$\{i\}$ & otherwise. \\%
\end{tabular}%
\right.
$$
Thus, $\pi_{\mathcal{E}} \gg \widetilde{\pi}$ via $\widetilde{C}.$ 
Now, it is possible to construct a sequence $C_1,\ldots,C_m$ of coalitions of $\mathcal{B}$, a sequence $\mathcal{E}_0,\ldots,\mathcal{E}_m$  of subsets of the maximal sets of  $\mathcal{B}$ such that $\mathcal{E}_0 =\mathcal{E}\cup \{\widetilde{C}\}$ and  $\mathcal{E}_m \subsetneq\widetilde{\mathcal{M}} $ with $\widetilde{C}$  having a non-empty intersection with  at least two coalitions in $\mathcal{E}_m$, and a sequence $\pi_{\mathcal{E}}=\pi_0,\ldots,\pi_m$ of coalition structures   fulfilling the following conditions for each $\ell=1,\ldots,m$:
\begin{enumerate}[(i)]
\item $C_{\ell} \in \mathcal{E}_{\ell} \subseteq \pi_{\ell},$
\item $\pi_\ell \gg \pi_{\ell-1}$ via $C_\ell$,
\item $|\mathcal{E}_\ell|=|\mathcal{E}|+1.$
\end{enumerate}
Thus, $\pi_m \gg^T \pi_{\mathcal{E}}.$ Let  $\mathcal{E}'\subsetneq \mathcal{E}_m $  be such that each coalition in $\mathcal{E}$ is disjoint with $\widetilde{C}$. Let $\widehat{\pi}$ be such that 
$$\widehat{\pi}(i)=\left\{ 
\begin{tabular}{ll}
$\widetilde{C}$ & $\text{for } i\in \widetilde{C}$\\
$\mathcal{E}'(i)$ & $\text{for } i\in N(\mathcal{E}')$ \\
$\pi(i)$ & $\text{for }i \in N \setminus  N(\mathcal{B})$\\
$\{i\}$ & otherwise. \\%
\end{tabular}%
\right.
$$
Thus, $\widehat{\pi} \gg \pi_m$ via $\widetilde{C}$. If $\mathcal{E}'=\emptyset$, then $\widehat{\pi}=\pi'$. By Definition \ref{def generalize ring}, $\pi^{\star} \gg^T \pi'$, so $\pi^{\star} \gg^T \pi'\gg \widetilde{\pi}\gg^T \pi$ implying $\pi^{\star} \gg^T \pi$, and the proof is complete.   If $\mathcal{E}'\neq\emptyset$, then  $\pi_{\mathcal{E}'}$ is defined similarly to $\pi_{\mathcal{E}}$ and  the same reasoning used when $\pi_{\mathcal{E}}$ was defined can be repeated. Eventually, we reach a coalition structure $\widehat{\pi}$  such that the only coalitions in common with $\widetilde{\mathcal{M}}$  are those that have non-empty intersection with $\widetilde{C}$. Then, such coalition structure is dominated by $\pi'$ via $\widetilde{C}$.
Thus, $\pi'\gg \widehat{\pi} \gg^T \pi_{\mathcal{E}}\gg \widetilde{\pi} \gg^T \pi.$ By Definition \ref{def generalize ring}, $\pi^{\star} \gg^T \pi'$. Therefore, $\pi^{\star} \gg^{T}\pi.$
\end{proof}
\begin{lemma}\label{compacta dominada} Let $\mathcalligra{R} \ $ \ be a reduced form and let  $\pi_{\mathcalligra{R}}$ \  be an $\mathcalligra{R}$ \ -- coalition structure. If $\pi\in\Pi$ is such that each non-singleton coalition in $\pi_{\mathcalligra{R}}$ \ belongs to  $\pi$,  then $\pi = \pi_{\mathcalligra{R}} \ $ or  $\pi_{\mathcalligra{R}} \gg^T \pi.$  
\end{lemma}
\begin{proof}
Let $\mathcalligra{R} \ $ \ be a reduced form and let $\pi_{\mathcalligra{R}}$ \  be an $\mathcalligra{R}$ \ -- coalition structure. Let  be $\mathcal{M}$  the set of all non-singleton coalitions in $\pi_{\mathcalligra{R}}$ \ and consider $\pi\in\Pi$  such that $\mathcal{M}\subseteq \pi$. 
There are two cases to consider:
\begin{enumerate}

\item [$\boldsymbol{1.}$] \textbf{Each non-singleton coalition $\boldsymbol{C \in \pi \setminus \mathcal{M}}$ belongs to a generalized ring or fixed component of $\boldsymbol{\mathcalligra{R}}$\ }. If $\mathcalligra{R}$ \  only includes generalized rings of type 1 or fixed components then, by Definition \ref{definicion de D coalition structure}, $\pi = \pi_{\mathcalligra{R}} \ $ and the proof is complete. Assume that there is only one non-compact generalized ring, say $\mathcal{B}$,  in $\mathcalligra{R}$\ . Thus, by Lemma  \ref{reducir a compact collection}, there are $C'\in \mathcal{B}$ and $\pi^{\star}\in \Pi$ with 
$$\pi^{\star} (i)=\left\{ 
\begin{tabular}{ll}
$C'$ & $\text{for } i\in C'$ \\
$\pi(i)$ & $\text{for }i \in N \setminus  N(\mathcal{B})$\\
$\{i\}$ & otherwise \\%
\end{tabular}%
\right.
$$
 such that $\pi^{\star}\gg^T \pi$. By Definition \ref{definicion de D coalition structure}, $ \pi^{\star}$  is an $\mathcalligra{R} \ $--coalition structure. Thus, by Lemma \ref{pisubde}, $\pi_{\mathcalligra{R}} \ \gg^T \pi^{\star}$ and the proof is complete. Assume now that there is more than one non-compact generalized ring. Applying Lemma \ref{reducir a compact collection} for each of these generalized rings, the proof follows similarly.

\item [$\boldsymbol{2.}$]\textbf{There is a non-singleton coalition $\boldsymbol{C \in \pi \setminus \mathcal{M}$ such that $C}$ does not  belong to } \textbf{any generalized ring or any fixed component of $\boldsymbol{\mathcalligra{R}}\ .$  } Thus, either $C$ belongs to a generalized ring or  it defines a fixed component (formed by agents in $\mathcal{S}$), that we denote by  $\widetilde{\mathcal{B}}$. In either case, by Definition \ref{reduced form}, $\widetilde{\mathcal{B}}$ is not protected by  $\mathcalligra{R}\ $. Next define $\pi^{\star} \in \Pi$ as follows.
If there is no non-compact generalized ring in $\mathcalligra{R}\ $, $\pi^{\star}=\pi$. If there are non-compact generalized rings in $\mathcalligra{R}\ $,  using Lemma \ref{reducir a compact collection} repeatedly we can construct $\pi^{\star} \in \Pi$ such that it contains only one coalition of each non-compact generalized ring, $\pi^{\star}(i)=\pi(i)$ for each agent $i$ not in any non-compact generalized ring, all remaining agents are singletons in $\pi^{\star}$, and $\pi^{\star} \gg^T \pi.$

Given a coalition structure $\pi \in \Pi$, let $|\pi|_\mathcal{K}$ denote the number of  coalitions $C\in \pi$ with $|C|>1.$ 

\noindent \textbf{Claim: there is a coalition structure $\boldsymbol{\widetilde{\pi}$ with set of non-singleton coalitions $\mathcal{M}''}$ } \textbf{   such that  $\boldsymbol{\mathcal{M} \subseteq \mathcal{M}''}$,}  $\boldsymbol{\widetilde{\pi} \gg^T \pi^{\star}}$, \textbf{and}  $\boldsymbol{|\widetilde{\pi}|_\mathcal{K} < |\pi^{\star}|_\mathcal{K}.}$

\noindent To prove the Claim, note that in $\pi^{\star}$ there is a non-singleton coalition with a non-empty intersection with agents in $N(\widetilde{\mathcal{B}}),$ say $C_0$. First, assume that  $C_0\subseteq N(\widetilde{\mathcal{B}})$. Thus, $C_0$ defines a fixed component or belongs to a generalized ring,  and Definition \ref{reduced form}  (ii) implies the existence of a coalition $C_1$ that breaks the fixed component or generalized ring to which $C_0$ belongs and nothing impedes $C_1$ from being formed.
If there is no fixed component and no generalized ring $\mathcal{B}\in \!\mathcalligra{R}\ $ such that   $C_1\cap N(\mathcal{B})\neq \emptyset$, then $C_1$ defines a fixed component or belongs to a generalized ring formed by agents in $N(\widetilde{\mathcal{B}}).$  In either case, Definition \ref{reduced form} implies the existence of a coalition $C_2$ that breaks the fixed component or the generalized ring to which $C_1$ belongs and nothing impedes  $C_2$ from being formed. If there is neither a fixed component nor a generalized ring  $\mathcal{B}\in \!\mathcalligra{R}\ $ such that   $C_2\cap N(\mathcal{B})\neq \emptyset$, repeat the previous argument. Continuing this reasoning,  it is possible to construct a sequence of coalitions $C_0, C_1,\ldots,C_m$ and a sequence of coalition structures $\pi_0, \pi_1,\ldots,\pi_{m-1}$ with $\pi_0=\pi^{\star}$  fulfilling the following conditions:
\begin{enumerate}[(i)]
\item $C_\ell \succ C_{\ell-1}$  for each $\ell=1,\ldots,m$;
\item $C_{\ell-1}\subseteq N(\widetilde{\mathcal{B}})$ and $\pi_\ell \gg \pi_{\ell-1}$ via $C_\ell$  for each $\ell=1,\ldots,m-1$; and

\item there is a fixed component or a generalized ring $\mathcal{B}'\in \!\mathcalligra{R}\ $ such that  $C_m \cap N(\mathcal{B}')\neq \emptyset.$
\end{enumerate}  Note that, by Conditions (i) and (ii) above, $C_m\cap N(\widetilde{\mathcal{B}})\neq \emptyset.$ Also, the existence of a fixed component or a generalized ring  $\mathcal{B}'$ in (iii) is ensured by the finiteness of the number of coalitions of the game and by Definition \ref{reduced form}. Given that $\mathcal{B}'\in \!\mathcalligra{R}\ $ is a generalized ring\footnote{ Otherwise, since $\mathcal{B}'$ does not impedes $C_m$ from being formed, $C_m$ would break $\mathcal{B}'.$ This contradicts  Definition \ref{reduced form}.} and does not impede $C_m$ from being formed, there is a maximal set of $\mathcal{B}'$ (when $\mathcal{B}'$ is of type 1) or a coalition in $\mathcal{B}'$ (when $\mathcal{B}'$ is non-compact) that has empty intersection with $C_m$. Use the term $\mathcal{E}$ for that maximal set (when $\mathcal{B}'$ is of type 1) or for the singleton that includes that coalition  (when $\mathcal{B}'$ is non-compact). 
By  Definition \ref{def generalize ring}, there is a coalition structure $\pi'$ such that $\pi'\gg^T\pi_{m-1}$ with $\pi'(i)=\pi_{m-1}(i)$ for each $i\in N \setminus N(\mathcal{B}')$, and either $\pi'(i)=\{i\}$ or $\pi'(i)\subseteq\mathcal{E}$ for each $i\in N(\mathcal{B}')$. Thus, there is $\pi_m$ such that $\pi_m \gg \pi'$ via $C_m.$ 
Given that $\mathcal{B}'$ is a generalized ring, $C_m\cap N(\mathcal{B}')\neq\emptyset$, and $\mathcal{B}'$ does not impede $C_m$ from being formed, there are $\mathcal{E}'$ (a maximal set when $\mathcal{B}'$ is of type 1 or a singleton that includes a coalition when $\mathcal{B}'$ is non-compact), a coalition $R_0\in \mathcal{E}'$ such that $R_0$ breaks $\mathcal{E}$ and  $R_0\succ C_m$, and a coalition structure $\pi''$ such that  $\pi''\gg\pi_m$ via $R_0$.\footnote{If $\mathcal{E}$ is the unique set such that $C_m\cap N(\mathcal{E})=\emptyset$, the existence of $R_0$ such that $R_0\succ C_m$ is guaranteed.  If it is not unique, it can be selected such that $\pi''\gg \pi_m$ via $R_0$.} 
Given that $\mathcal{E}'$ may not be  included in  $\pi^{\star}$, and $\mathcal{B}'$ is a generalized ring, there are a sequence of coalitions  $R_1,\ldots,R_s$ in  $\mathcal{B}',$ and a sequence of coalition structures $\widetilde{\pi}_0, \widetilde{\pi}_1,\ldots,\widetilde{\pi}_s$ with $\widetilde{\pi}_0=\pi''$  fulfilling the following conditions for each $\ell=1,\ldots,s$:
\begin{enumerate}[(i)]
\item $R_\ell \succ R_{\ell-1}$,
\item $\widetilde{\pi}_\ell \gg \widetilde{\pi}_{\ell-1}$ via $R_\ell$, and  
\item $\widetilde{\pi}_s(i)= \pi^{\star}(i)$ for each $i\in N(\mathcal{B}')$.
\end{enumerate}
Let  $\widetilde{\pi} \equiv \widetilde{\pi}_s$. Note that, since $C_m \in \pi_m$,  $R_0\succ C_m$, $R_0$ breaks $\mathcal{E}$, and $\mathcal{E}\subseteq \pi_m$, it follows that   $|\pi_m|_{\mathcal{K}}>| \widetilde{\pi}|_{\mathcal{K}}.$  Also, by construction of the sequence $\pi_0,\ldots,\pi_m$, $|\pi^{\star}|_{\mathcal{K}} \geq | \pi_m|_{\mathcal{K}}$. Therefore, $|\pi^{\star}|_{\mathcal{K}} > | \widetilde{\pi}|_{\mathcal{K}}.$ Moreover, by Condition (iii) verified by the sequence $\widetilde{\pi}_0, \widetilde{\pi}_1,\ldots,\widetilde{\pi}_s$, we have that $\mathcal{M} \subseteq \mathcal{M}''$. Thus, $\widetilde{\pi}$ fulfills the conditions of the Claim and the Claim holds when $C_0\subseteq \widetilde{\mathcal{B}}$. Second, assume that $C_0\cap \left(N \setminus N(\widetilde{\mathcal{B}})\right)\neq \emptyset$. Thus,  $C_0$ can be considered as the coalition $C_m$ in the previous reasoning, and the proof follows similarly.  This completes the proof of the Claim.

Now, we conclude the proof of Lemma \ref{compacta dominada}. First note that, by the Claim, a coalition structure $\widetilde{\pi}$ is obtained such that $\widetilde{\pi}\gg^T \pi^{\star},$   $\mathcal{M} \subseteq \mathcal{M}''$,  and $|\widetilde{\pi}|_\mathcal{K} < |\pi^{\star}|_\mathcal{K}.$ If $\widetilde{\pi}=\pi_{\mathcalligra{R}} \ $, then  the proof is complete. Otherwise, there is $C' \in \mathcal{M}'' \setminus \mathcal{M}$ such that $C'$ does not  belong to any fixed component or generalized ring of $\mathcalligra{R}\ .$ By applying the Claim to coalition structure $\widetilde{\pi}$ we obtain a new coalition structure $\widehat{\pi}$ such that $\widehat{\pi} \gg^T \widetilde{\pi},$ $\mathcal{M} \subseteq \mathcal{M}''$  and $|\widehat{\pi}|_\mathcal{K} < |\widetilde{\pi}|_\mathcal{K} < |\pi^{\star}|_\mathcal{K} .$  If $\widehat{\pi}=\pi_{\mathcalligra{R}} \ $, then  the proof is completed. Otherwise, continue applying the Claim until coalition structure $\pi_{\mathcalligra{R}} \ $ is obtained.
\end{enumerate}\end{proof}

\begin{lemma}\label{pisubde bis} 
Let $\mathcalligra{R}$ \  \ be a reduced form, let $\pi_{ \!\! \mathcalligra{R}} \ $ be an $\mathcalligra{R}$ -- coalition structure and let $\pi\in \Pi$. If $\pi \gg^T \pi_{ \!\! \mathcalligra{R}} \ $, then $\pi_{ \!\! \mathcalligra{R}} \ \gg^T \pi.$
\end{lemma}
\begin{proof}
Let $\mathcalligra{R}$ \  \ be a reduced form, let $\pi_{ \!\! \mathcalligra{R}} \ $ be an $\mathcalligra{R}$ -- coalition structure and let $\pi\in \Pi$ such that $\pi \gg^T \pi_{ \!\! \mathcalligra{R}} \ $. By Definition \ref{reduced form}, there is an $\mathcalligra{R}$ -- coalition structure $\pi'_{ \!\! \mathcalligra{R}} \ $ with  set of non-singleton coalitions  $\mathcal{M}'$ such that $\mathcal{M}' \subseteq \pi.$ Thus, by Lemma \ref{compacta dominada}, $\pi'_{\!\!\mathcalligra{R}} \ =\pi$ or  $\pi'_{\! \!\mathcalligra{R}}\  \gg^T \pi.$ Then, by Lemma \ref{pisubde}, $\pi_{\!\!\mathcalligra{R}} \ \gg^T \pi'_{\mathcalligra{R}}$ \ and, therefore, $\pi_{\!\!\mathcalligra{R}} \ \gg^T \pi.$
\end{proof}
\bigskip

 Given  a reduced form $\mathcalligra{R}$ \ \ and an $\mathcalligra{R}\ $ -- coalition structure $\pi_{\mathcalligra{R}} \ ,$  the \textit{set generated by} $\pi_{\mathcalligra{R}} \ ,$ denoted by $\mathcal{A}_{\!\mathcalligra{R}}$\ ,   is the set formed by  $\pi_{\mathcalligra{R}} \ $ together with all the coalition structures that transitively dominate it. Formally,  $$\mathcal{A}_{\!\mathcalligra{R}} \equiv \{\pi_{\mathcalligra{R}} \ \}\cup \{\pi \in \Pi : \pi \gg^T \pi_{\mathcalligra{R}} \ \}.$$ 
The following result states that the set $\mathcal{A}_{\!\mathcalligra{R}} \ $ is actually an absorbing set.
 
\begin{proposition}\label{absorbing and stable decomp}
If $\mathcalligra{R}$ \  \ is a reduced form, then $\mathcal{A}_{\!\mathcalligra{R}}$ \ \ is an absorbing set. Furthermore, if there is a (compact or non-compact) generalized ring in $\mathcalligra{R}$  ,  $\mathcal{A}_{\!\mathcalligra{R}}$ \ \ is a non-trivial absorbing set.
\end{proposition}
\begin{proof}
Let $\mathcalligra{R}\ =\{\mathcal{B}_1, \ldots, \mathcal{B}_L,\mathcal{S}\}$ \   be a reduced form. 
First, assume that $\mathcal{B}_\ell=\{C_\ell\}$ for each $\ell=1,\ldots,L.$  Thus, $\pi_{\mathcalligra{R}} \  \ $ is a stable coalition structure. If not, there is a coalition that blocks the coalition structure $\pi_{\mathcalligra{R}} \  \ $. Since $\mathcalligra{R}\ $ is a  reduced form, such a blocking coalition must be formed by agents in $\mathcal{S}$, and it must belong to a fixed component or to a generalized ring, say $\mathcal{B}'$, that is not protected by $\mathcalligra{R}\ $. Therefore,   there is another blocking coalition formed by agents in $\mathcal{S}$ that blocks $\mathcal{B}'$. Repeating this reasoning and by the finiteness of the set of permissible coalitions, we eventually reach a fixed component or a generalized ring formed by agents in $\mathcal{S}$ which is not blocked, generating a contradiction since $\mathcalligra{R}\ $ is a reduced form. Hence, $\pi_{\mathcalligra{R}} \ \ $ is stable   and, by Remark \ref{remarkabsorbing} (i), $\mathcal{A}_{\!\mathcalligra{R}} \ =\{\pi_{\mathcalligra{R}} \ \}$ is an absorbing set. Second, assume  there is $\mathcal{B} \in \!\mathcalligra{R} \ \ $ such that  $|\mathcal{B} |\geq 3.$ Thus, $\mathcal{B}$ is a generalized ring of $\!\mathcalligra{R} \ $. By Definition \ref{def generalize ring}, Definition \ref{reduced form}, and Definition \ref{definicion de D coalition structure}, there are at least three different $\mathcalligra{R}\ $ -- coalition structures.  Thus, by Lemma \ref{pisubde},  $|\mathcal{A}_{\!\mathcalligra{R}} \ |\geq 3.$ Let $\pi \in \mathcal{A}_{\!\mathcalligra{R}}$ \ and consider $\pi' \in \Pi \setminus\{\pi\}$ such that $\pi' \gg^T \pi.$ As $\pi \in \mathcal{A}_{\!\mathcalligra{R}} \ ,$ it follows that $\pi \gg^T \pi_{\mathcalligra{R}}\ .$ Therefore, by transitivity of $\gg^T$, $\pi' \gg^T \pi_{\mathcalligra{R}} \ $ and $\pi' \in \mathcal{A}_{\!\mathcalligra{R}} \ .$ Next, let $\pi$ and $\pi'$ be different coalition structures of  $\mathcal{A}_{\!\mathcalligra{R}} \ .$ By definition of $\mathcal{A}_{\!\mathcalligra{R}} \ ,$ it follows that  $\pi \gg ^T \pi_{\mathcalligra{R}}$ \ and $\pi' \gg ^T \pi_{\mathcalligra{R}} \ .$   By Lemma \ref{pisubde bis}, it follows that $ \pi_{\mathcalligra{R}} \ \gg ^T \pi.$ Thus, $\pi' \gg ^T \pi_{\mathcalligra{R}} \ $ together with $ \pi_{\mathcalligra{R}} \ \gg ^T \pi$ and the transitivity of $\gg^T$ imply $\pi' \gg^T \pi.$  This proves that $\mathcal{A}_{\!\mathcalligra{R}} \ $ satisfies the conditions of Definition \ref{absorbing}. Therefore, $\mathcal{A}_{\!\mathcalligra{R}} \ $ is a non-trivial absorbing set.
\end{proof}

\medskip

 Recall that, by the definition of a compact generalized ring, each coalition that breaks a maximal set has a non-empty intersection with only one coalition of that maximal set. Thus,  when a reduced form includes a compact generalized ring, each coalition structure of the absorbing set generated includes a maximal set of that generalized ring. By contrast, when a reduced form includes a non-compact generalized ring there are coalition structures in the generated absorbing set that contain only one coalition of such generalized ring (and the remaining agents involved in the said generalized ring appear as singletons in those coalition structures). 

Note that, by Lemma \ref{pisubde}, the absorbing set $\mathcal{A}_{\!\mathcalligra{R}} \ $ depends only on the reduced form $\mathcalligra{R}\ $, and not on the specific $\mathcalligra{R}$ -- coalition structure selected to construct it.

Finally, we are in a position to prove the main result of the paper. \bigskip

\noindent\begin{proof}[Proof of Theorem \ref{bijection}]To prove that there is a bijection between  reduced forms and absorbing sets for each coalition formation game, we show that  $\mathcal{A}$ is an absorbing set if and only if  $\mathcal{A}=\mathcal{A}_{\mathcalligra{R}}$ \ \ for a reduced form $\mathcalligra{R} \ .$ The proof  ($\Longleftarrow$)  follows from Proposition \ref{absorbing and stable decomp}. To prove ($\Longrightarrow$), let $\mathcal{A}$ be an absorbing set. If $|\mathcal{A}|=1,$ by Remark \ref{remarkabsorbing} (ii), the unique element of $\mathcal{A}$ is a stable coalition structure. Thus, by Remark \ref{remark reduced form con stable partition},  a reduced form can be induced. If $|\mathcal{A}|>1,$ the reduced form $\mathcalligra{R} \ $ is constructed as follows. First, by Lemma \ref{construyo todos los RC} of Appendix \ref{section ring en preferences},   each of the generalized rings involved in $\mathcal{A}$ can be identified. Let $\mathcalligra{B} \ \ $ be the collection of all those generalized rings. Next, consider the subcollection $$\mathcalligra{B} \  \ ^\star =\{\mathcal{B} \in \mathcalligra{B} \ \ : \text{ for each } \pi \in  \mathcal{A} \text{ there is }C \in \mathcal{B} \text{ such that }C \in \pi\}$$
and include each generalized ring of $\mathcalligra{B} \  ^\star$ as an element of $\mathcalligra{R} \ .$ Second, let $\mathcal{F}=\{C \in \mathcal{K} : C \in \pi \text{  for each  }\pi \in \mathcal{A}\}$ and for each $C \in \mathcal{F}$ include $\{C\}$ as an element of $\mathcalligra{R} \ .$  Now, it remains to be shown that $\mathcalligra{R} \ $ is a reduced form.  To do this, take any $\mathcal{B}\in\!\!\mathcalligra{R} \  $.  The goal is to confirm that $\mathcal{B}$ is protected by $\!\!\mathcalligra{R} \  $.  Assume otherwise. Thus, there is $C \in \mathcal{K}$ that breaks $\mathcal{B}$ and no other element of $\!\!\mathcalligra{R} \ \ $  impedes $C$ from being  formed. This implies that there are $\pi, \pi' \in \mathcal{A}$ and $C' \in \mathcal{B}$ such that $\pi' \gg \pi$ via $C,$ $C' \in \pi$ and $C \succ C'.$ By the definition of absorbing set it also follows that $\pi \gg^T \pi'.$  Thus, there is a cycle $\mathcal{C}$ of coalition structures that contains $\pi$ and $\pi'.$  By Proposition \ref{Th ring cycle} in Appendix \ref{section ring en preferences}, cycle $\mathcal{C}$ induces a ring of coalitions that contains both coalitions $C'$ and $C.$ Therefore, $C$ belongs to generalized ring $ \mathcal{B},$ which is absurd since $C$ breaks $\mathcal{B}$. Hence, $\mathcal{B}$ is protected by  $\!\!\mathcalligra{R} \ \ $  and Condition (i) of Definition \ref{reduced form} holds. To see that Condition (ii) of Definition \ref{reduced form} holds, let  $\mathcal{B}$ be a generalized ring or a fixed component formed by agents in $\mathcal{S}$.  We need to show that $\mathcal{B}$ is not protected by $\mathcalligra{R} \ .$  Assume otherwise. There are two cases to consider:
\begin{enumerate}
\item [$\boldsymbol{1.}$] \textbf{No coalition breaks $\boldsymbol{{\mathcal{B}}}$}. Thus, by construction of $\!\!\mathcalligra{R} \ $, $\mathcal{B}$   belongs either to $\mathcalligra{B} \  \ ^\star$ or to $\mathcal{F}$, contradicting the requirement that $\mathcal{B}$ is formed by agents in $\mathcal{S}$. 

\item [$\boldsymbol{2.}$] \textbf{There is a coalition $\boldsymbol{C$ that breaks $\mathcal{B}$ and   $\mathcal{B}$ is protected by $\!\!\mathcalligra{R} \ }$ }.
Let $\pi_{\!\!\mathcalligra{R}} \ $ be a  $\! \! \mathcalligra{R}\ $-- coalition structure. Thus $\pi_{\!\!\mathcalligra{R}} \ (i)=\{i\}$ for each $i\in N(\mathcal{B})$ (since $N(\mathcal{B}) \subseteq \mathcal{S}$). Since $C$ breaks $\mathcal{B}$ and $\mathcal{B}$ is protected by $\! \! \mathcalligra{R}\ $, and the fact that Condition (i) holds, $C\subseteq N(\mathcal{B}).$ Thus, there is $\pi\in \mathcal{A}$ such that  $\pi \gg \pi_{\!\!\mathcalligra{R}} \ $ via $C.$ Therefore, $\mathcal{B}$  belongs either to $\mathcalligra{B} \  \ ^\star$ or to $\mathcal{F}$, contradicting the requirement that $\mathcal{B}$ is formed by agents in $\mathcal{S}$. 

\end{enumerate}
Therefore,  $\mathcalligra{R} \ \ $  is a reduced form.
\end{proof}

\end{document}